\documentclass{article}

\usepackage[T1]{fontenc}
\usepackage{amssymb, amsmath, graphicx, subfigure, xparse}
\usepackage{physics, authblk}
\usepackage{tikz, tkz-graph}
\usepackage{fullpage}
\usepackage{thm-restate}
\usepackage{comment}
\usepackage{multirow}
\usepackage{qcircuit}

\newtheorem{theorem}{Theorem}
\newtheorem{definition}[theorem]{Definition}

\newtheorem{lemma}[theorem]{Lemma}

\newtheorem{fact}[theorem]{Fact}

\newtheorem{conjecture}[theorem]{Conjecture}

\newtheorem{innercustomthm}{Theorem}
\newenvironment{customthm}[1]
  {\renewcommand\theinnercustomthm{#1}\innercustomthm}
  {\endinnercustomthm}
  
\newtheorem{innercustomconj}{Conjecture}
\newenvironment{customconj}[1]
  {\renewcommand\theinnercustomconj{#1}\innercustomconj}
  {\endinnercustomconj}

\newenvironment{proof}{\noindent{\bf Proof:} \hspace*{1mm}}{
    \hspace*{\fill} $\Box$ }

\usepackage{amsmath, amssymb, graphicx, subfigure}
\usepackage{enumitem}
\usepackage{xparse}
\usepackage{physics}

\newcommand{\x}{\mathbf{x}}

\newcommand{\FF}{\mathbb{F}}

\newcommand{\Aa}{\mathcal{A}}

\newcommand{\Dd}{\mathcal{D}}

\newcommand{\Hh}{\mathcal{H}}

\newcommand{\Jj}{\mathcal{J}}
\newcommand{\Oo}{\mathcal{O}}

\newcommand{\Uu}{\mathcal{U}}

\newcommand{\perm}{\mathsf{perm}}

\newcommand{\eps}{\epsilon}

\newcommand{\defeq}{\mathrel{\overset{\makebox[0pt]{\mbox{\normalfont\tiny\sffamily def}}}{=}}}

\usepackage{complexity}
\newcommand{\SHARPP}{\mathsf{\#P}}
\newcommand{\Poly}{\mathsf{P}}

\newcommand{\bits}{\{0,1\}}

\newtheorem{cor}[theorem]{Corollary}

\newcommand{\fn}[3]{#1: #2 \rightarrow #3}
\newcommand{\bitsfn}[3]{\fn{#1}{\bits^{#2}}{\bits^{#3}}}

\usepackage{pgfplots}

\newcommand{\ce}{\mathsf{CE}}
\newcommand{\ced}{\mathsf{CED}}

\newcommand{\trunctaylor}{truncated perturbed Haar-distribution}
\newcommand{\taylor}{perturbed Haar-distribution}
\newcommand{\haar}{Haar-random circuit distribution}
\newcommand{\Done}{\mathcal{D}}
\newcommand{\Dtwo}{\mathcal{D}'}
\newcommand{\przero}[1]{{\sf{p}}_{\bf{0}}(#1)}

\newcommand{\Random}{\mathcal{H}}
\newcommand{\arXiv}[1]{#1}
\newcommand{\natphys}[1]{}

\begin{document}

\title{Quantum Supremacy and the Complexity of \\ Random Circuit Sampling}
\author[1]{Adam Bouland}
\author[1,2]{Bill Fefferman\thanks{Corresponding Author. e-mail: $\mathtt{wjf@berkeley.edu}$.}}
\author[1]{Chinmay Nirkhe}
\author[1]{Umesh Vazirani}
\affil[1]{\small \textit{Department of Electrical Engineering and Computer Sciences, \protect\\ University of California, Berkeley}}
\affil[2]{\small \textit{Joint Center for Quantum Information and Computer Science (QuICS), \protect\\ University of Maryland / NIST}}
\date{}

\maketitle

\begin{abstract}
\natphys{
A critical milestone on the path to useful quantum computers is quantum supremacy -- a demonstration of a quantum computation that is prohibitively hard for classical computers. A leading near-term candidate, put forth by the Google/UCSB team, is sampling from the probability distributions of randomly chosen quantum circuits, which we call Random Circuit Sampling (RCS).

We establish the first complexity-theoretic evidence of classical hardness of RCS, placing it on par with the best theoretical proposals for supremacy. Specifically, we show that RCS satisfies an average-case hardness condition -- computing most output probabilities is as hard as computing all of them. In addition, it follows from known results that RCS also satisfies an anti-concentration property, making it the first proposal with both.

Ideally, one would like to tie this complexity-theoretic evidence %
to the actual statistical tests used to verify experimental devices. 
We give several examples which clarify necessary properties for such statistical tests, and in particular, show that the leading proposal for such a test -- cross-entropy -- does not make this connection directly.
}
\arXiv{
A critical milestone on the path to useful quantum computers is quantum supremacy -- a demonstration of a quantum computation that is prohibitively hard for classical computers. A leading near-term candidate, put forth by the Google/UCSB team, is sampling from the probability distributions of randomly chosen quantum circuits, which we call Random Circuit Sampling (RCS).

In this paper we study both the hardness and verification of RCS.  While RCS was defined with experimental realization in mind, we show complexity theoretic evidence of hardness that is on par with the strongest theoretical proposals for supremacy.
Specifically, we show that RCS satisfies an average-case hardness condition -- computing output probabilities of typical quantum circuits is as hard as computing them in the worst-case, and therefore $\SHARPP$-hard.  Our reduction exploits the polynomial structure in the output amplitudes of random quantum circuits, enabled by the Feynman path integral. In addition, it follows from known results that RCS satisfies an anti-concentration property, making it the first supremacy proposal with both average-case hardness and anti-concentration.
}

 \end{abstract}

\section{Introduction}

In the early 1990's, complexity-theoretic techniques provided the first theoretical demonstrations that quantum computers have the potential to solve certain computational problems exponentially faster than classical computers \cite{bvSTOC,simonFOCS}.  These paved the way for remarkable results showing that fully fault-tolerant, scalable quantum computers will be able to quickly factor large integers \cite{shor1999polynomial}, as well as simulate quantum mechanical systems \cite{feynman1982simulating,lloyd1996universal}. While quantum devices capable of solving such important problems may still be far off, decades of work undertaken toward building scalable quantum computers have already yielded considerable progress in high-precision control over quantum systems. Indeed, at present, several concurrent experimental efforts from groups in industry and academia such as Google, IBM, and the University of Maryland have already reached the point where systems of around 50 high-quality qubits are within experimental reach \cite{naturenewscommercialize,ibm6qubits,zhang2017observation}.
 
As we approach this so-called Noisy Intermediate Scale era of Quantum computing (or ``NISQ'' \cite{1801.00862}), a key milestone will be \emph{quantum supremacy}: the quest to perform a computational task that can be solved by these systems but cannot be solved in a reasonable amount of time by any classical means. Akin to the early demonstrations of the power of quantum computers, there is no requirement that the computational task be useful -- the main additional requirement is that the task should be physically realizable in the near term. 

Broadly speaking, we can classify quantum supremacy proposals into two categories -- those seeking to provide very strong complexity-theoretic evidence of classical intractability while hoping to be physically realized in the near term, versus those with excellent prospects for physical realization in the short term while providing weaker evidence of classical intractability. This paper shows that these categories intersect by providing strong complexity-theoretic evidence of classical intractability for the leading candidate from the latter category. 

More specifically, the first category of quantum supremacy proposals had their origins in the desire to obtain strong complexity-theoretic evidence of the power of quantum computers. A  key insight was that focusing on the probability distributions quantum devices can sample from, rather than more standard notions of computing or optimizing functions, opens up the possibility of strong evidence of classical intractability. This perspective led to proposals such as BosonSampling \cite{aaronsonboson} and IQP \cite{BJS2010}, together with proofs that the probabilities of particular quantum outcomes correspond to quantities which are difficult to compute  -- such as matrix permanents. This allowed them to connect the hardness of classical simulation of such systems to well-supported hardness assumptions stemming from complexity theory. 

As an added bonus, Aaronson and Arkhipov realized that BosonSampling might be experimentally feasible in the near term, and helped jump-start the quest for quantum supremacy more than half a decade ago \cite{spring2012boson,broome2013photonic,tillmann2013experimental,crespi2013integrated}. More recently,  BosonSampling experiments have faced experimental difficulties with photon generation and detector efficiency, making it challenging to push these experiments to the scale required to achieve supremacy ($\sim50$ photons) \cite{neville2017no,clifford2018classical}. %
It remains to be seen if such experiments can be implemented in the near future.
  
The second category of supremacy results is directly inspired by the dramatic experimental progress in building high-quality superconducting qubits (e.g., \cite{boixo2016characterizing,naturenewscommercialize}).  These groups defined the natural sampling task for their experimental context, which we call \emph{Random Circuit Sampling} (RCS). The task is to take a random (efficient) quantum circuit of a specific form and generate samples from its output distribution. While RCS lacks some of the complexity-theoretic evidence that made BosonSampling so theoretically compelling, this proposal promises to be more readily scaled to larger system sizes in the near term. In particular, the group at Google/UCSB plans to conduct such an experiment on a 2D array of 49 qubits by the end of 2018 \cite{martinistalk}. 
 
Our main result gives strong complexity-theoretic support to this experimentally driven proposal.  In technical terms, this involves developing the first \emph{worst-to-average-case} reduction for computing the output probabilities of random quantum circuits. That is, we prove that the ability to compute the output probability of a typical quantum circuit is as hard as computing the probability of a worst-case circuit\footnote{To be more precise our reduction will work with respect to a natural discretized analog of the Haar measure.}. This is a necessary requirement to show hardness for random circuit sampling along the same lines as BosonSampling. Taken in combination with recent results establishing a subsequent piece of evidence known as \emph{anti-concentration} for these systems \cite{brandao2013exponential,hangleiter2017anti}, this puts RCS on par with the strongest theoretical proposals for supremacy.

Ideally, one would like to tie this complexity-theoretic evidence directly to the actual statistical tests used to verify experimental devices. %
One might hope that the leading candidate measure for verifying RCS, cross-entropy, would certify closeness in total variation distance\footnote{Specifically, cross-entropy is closely related to KL divergence, which is a known upper bound on total variation distance.}, the metric needed for the arguments above.
Unfortunately, there are simple counterexample distributions that score well on cross-entropy yet are far from ideal in total variation distance. In Section \ref{sec:our-results2}, we highlight these and other examples that help clarify some of the challenges in tying such statistical tests to complexity theoretic evidence. We note that this remains an open question for any supremacy proposal, including BosonSampling. 

%

%
 
%
%

\section{Our results: average-case hardness}
\label{sec:ourresults1}

Proposals for quantum supremacy have a common framework. The computational task is to sample from the output distribution $D$ of some experimentally feasible quantum process or algorithm (on some given input). To establish quantum supremacy we must show
\begin{enumerate}
    \item \emph{Hardness}:  no efficient classical algorithm can sample from any distribution close to $D$, and
    \item \emph{Verification}: an algorithm can check that the experimental device sampled from an output distribution close to $D$. 
\end{enumerate}
This need for verifiability effectively imposes a robustness condition on the difficulty of sampling from $D$. For example, the ability to sample one particular output $x$ of a quantum circuit with the correct probability $D(x)$ is known to be hard for classical computers, under standard complexity assumptions, e.g. \cite{terhaldivincenzo, BJS2010, morimae2014hardness,farhi2016quantum,boulandCCC2016}. But this is not a convincing proof of supremacy -- for one, under any reasonable noise model, this single output probability $D(x)$ might not be preserved.  Moreover, this single output $x$ is exponentially unlikely to occur -- and would therefore be extremely difficult to verify. Accordingly, any convincing proof of quantum supremacy must establish that $D$ is actually uniformly difficult to sample from. That is, the computational difficulty must be embedded across the entire distribution, rather than concentrated in a single output. 

The starting point for the BosonSampling proposal of Aaronson and Arkhipov consisted of three observations:  (1)  In general, for sufficiently hard problems (think $\SHARPP$-hard), showing a distribution $D$ is uniformly difficult to sample from corresponds to showing that for most outputs $x$, it is hard to compute $D(x)$. In complexity theory, this is referred to as ``average-case'' rather than ``worst-case'' hardness. (2) The output probabilities of systems of noninteracting bosons can be expressed as permanents of certain $n \times n$ matrices. (3) By a celebrated result of Lipton \cite{lipton_1991}, computing permanents of random matrices is $\SHARPP$-hard, or truly intractable in the complexity theory pantheon.  Together, these gave convincing evidence of the hardness of sampling typical outputs of a suitable system of noninteracting bosons, which could be experimentally feasible in the near term.

Specifically they proved that no classical computer can sample from any distribution within inverse-exponential total variation distance of the ideal BosonSampling output distribution.  Formally extending these results to experimentally relevant noise models, such as constant total variation distance, seems to require approximation robust worst-to-average-case reductions that are beyond the reach of current methods.  Nevertheless, their results, combined with the average-case hardness of the permanent, provide compelling evidence that BosonSampling has such robust hardness.

Permanents have a special structure enabling their average-case hardness -- an ingredient which is thus far missing in other supremacy proposals. Technically, average-case hardness is established by creating a ``worst-to-average-case reduction''. We will show such a reduction for RCS. At a high level, this involves showing that the value on a worst-case instance $x$ can be efficiently inferred from the values at a few random instances $r_1 , \ldots , r_m$. For this to be possible at all, while the $r_k$ might be individually random, their correlations must depend upon $x$ (which we shall denote by $r_0$). Typically such reductions rely on a deep global structure of the problem, which makes it possible to write the value at $r_k$ as a polynomial in $k$ of degree at most $m$. For example, the average-case property of permanents is enabled by its algebraic structure -- the permanent of an $n \times n$ matrix can be expressed as a low degree polynomial in its entries. This allows the value at $r_0 = x$ to be computed from the values at $r_k$ by polynomial interpolation. 

An astute reader may have noticed that randomizing the instance of permanent corresponds to starting with a random linear-optical network for the BosonSampling experiment, but still focusing on a fixed output. Our goal however was to show for a fixed experiment that it is hard to calculate the probability of a random output. These are equivalent by a technique called ``hiding''. By the same token, it suffices for RCS to show that it is hard to compute the probability of a fixed output, $0$, for a random circuit $C$.

To show this average-case hardness for quantum circuits, we start with the observation that the probability with which a quantum circuit outputs a fixed outcome $x$ can be expressed as a low degree multivariate polynomial in the parameters describing the gates of the circuit, thanks to writing the acceptance probability as a Feynman path integral. Unfortunately, this polynomial representation of the output probability does not immediately yield a worst-to-average-case reduction. At its core, the difficulty is that we are not looking for structure in an individual instance -- such as the polynomial description of the output probability for a particular circuit above. Rather, we would like to say that several instances of the problem are connected in some way, specifically by showing that the outputs of several different related circuits are described by a low degree (univariate) polynomial. With permanents, this connection is established using the linear structure of matrices, but quantum circuits do not have a linear structure, i.e. if $A$ and $B$ are unitary matrices, then $A+B$ is not necessarily unitary. This limitation means one cannot directly adapt the proof of average-case hardness for the permanent to make use of the Feynman path integral polynomial.

Here is a more sophisticated attempt to connect up the outputs of different circuits with a polynomial: Suppose we take a worst-case circuit $G = G_m \ldots G_1$, and multiply each gate $G_j$ by a Haar-random matrix $H_j$. By the invariance of the Haar measure, this is another random circuit -- it is now totally scrambled. Now we invoke a unique feature of quantum computation, which is that it is possible to implement a fraction of a quantum gate.  This allows us to replace each gate $H_j$ with $H_j e^{-i\theta h_j}$, where $h_j = - i \log H_j$ and $\theta$ is a small angle, resulting in a new circuit $G(\theta)$.  If $\theta=1$ this gives us back the worst-case circuit $G(1) = G$, but if $\theta$ is very tiny the resulting circuit looks almost uniformly random. One might now hope to interpolate the value of $G(1)$ from the values of $G(\theta_k)$ for many small values of $\theta_k$, thus effecting a worst-to-average reduction. Unfortunately, this doesn't work either. The problem is that $e^{-i\theta h_j}$ is not a low degree polynomial in $\theta$, and therefore neither is $G(\theta)$, so we have nothing to interpolate with.

The third attempt, which works, is to consider using a truncated Taylor series of $e^{-i\theta h_j}$ in place of $e^{-i\theta h_j}$ in the above construction. The values of the probabilities  in this truncation will be very close to those of the proposal above -- and yet by construction we have ensured our output probabilities are low degree polynomials in theta. Therefore, if we could compute most output probabilities of these "truncated Taylor" relaxations of random circuits, then we could compute a worst-case probability. 

\begin{theorem}[Simplified]
It is $\SHARPP$-hard to compute $\abs{\mel{0}{C'}{0}}^2$ with probability $8/9$ over the choice of $C'$, where $C'$ is drawn from any one of a family of discretizations of the Haar measure.
\label{averagecasethm}
\end{theorem}

These truncated circuit probabilities are slightly different from the average-case circuit probabilities but are exponentially close to them (even in relative terms). However, they are essentially the same from the perspective of supremacy arguments because quantum supremacy requires that computing most output probabilities even approximately remains $\SHARPP$-hard, and our perturbations to the random circuits fall within this approximation window. Therefore we have established a form of worst-to average-case reduction which is necessary, but not sufficient, for the supremacy condition to remain true. This is directly analogous to the case of permanents, where we know that computing average-case permanents exactly is $\SHARPP$-hard, but we do not know this reduction is sufficiently robust to achieve quantum supremacy.

RCS does satisfy an additional robustness property known as ``anti-concentration''. Anti-concentration states that the output distribution of a random quantum circuit is ``spread out''  --  that most output probabilities are reasonably large. Therefore, any approximation errors in estimating these probabilities are small relative to the size of the probability being computed. Once one has established a worst-to-average-case reduction, anti-concentration implies that there is some hope for making this reduction robust to noise -- intuitively it says that the signal is large compared to the noise.

Of the numerous quantum supremacy proposals to date which are robust to noise
\cite{aaronsonboson,feffermanfourier,bremner2016average,boixo2016characterizing,aaronson2016complexity,bremner2016achieving,morimae2017hardness,hangleiter2017anti,boulandcccs,bremner2017upcoming}, 
only two have known worst-to-average-case reductions: BosonSampling and FourierSampling \cite{aaronsonboson,feffermanfourier}. However, it remains open if these proposals also anti-concentrate.  On the other hand, many supremacy proposals have known anti-concentration theorems (see e.g., \cite{bremner2016average,boixo2016characterizing,bremner2016achieving,morimae2017hardness,hangleiter2017anti,boulandcccs,bremner2017upcoming}), but lack worst-to-average-case reductions. We note, however, that anti-concentration is arguably less important than worst-to-average case reductions, as the latter is necessary for quantum supremacy arguments, while the former is not expected to be necessary.  In the case of RCS, anti-concentration follows from prior work \cite{brandao2013exponential,hangleiter2017anti}. Therefore, our work is the first to show that both can be achieved simultaneously.  

\newcommand*\circled[1]{\tikz[baseline=(char.base)]{
            \node[shape=circle,draw,inner sep=1pt] (char) {#1};}}

\begin{figure}[h!]
\scriptsize
\centering
\begin{tabular}{l | c | c | c | c}
\multirow{2}{4em}{\textit{Proposal}} & \textit{Worst-case} & \textit{Average-case} & \multirow{2}{9em}{\textit{Anti-concentration}} & \textit{Imminent} \\
& \textit{hardness} & \textit{hardness} & & \textit{experiment} \\ \hline
BosonSampling \cite{aaronsonboson} & \checkmark & \checkmark & & \\
FourierSampling \cite{feffermanfourier} & \checkmark & \checkmark & & \\
IQP \cite{BJS2010,bremner2016average,bremner2016achieving} & \checkmark & & \checkmark &  \\
{\bf RCS \cite{brandao2013exponential,boixo2016characterizing,boixo2017fourier,neill2017blueprint,hangleiter2017anti}} & {\checkmark} & {\circled{\checkmark}} & {\checkmark} & {\checkmark}
\end{tabular}
\caption{The leading quantum supremacy proposals.}
\end{figure} 

\section{Our results: statistical verification of Random Circuit Sampling}
\label{sec:our-results2}

We now turn to verifying that an experimental realization of Random Circuit Sampling has performed RCS faithfully.  Verification turns out to be quite challenging. The first difficulty is that computing individual output probabilities of an ideal quantum circuit requires exponential classical time. However, current proposals leverage the fact that $n=50$ is small enough that it is feasible to perform this task on a classical supercomputer. The second difficulty is that one can only take a small number of samples from the experimental quantum device. This means there is no hope of experimentally observing all $2^{50}$ outcomes, nor of estimating their probabilities empirically\footnote{Nor of performing complete tomography, at this would both require a large number of measurements, and moreover would require one to trust the very quantum operations one is seeking to verify.}. The challenge is therefore to develop a statistical measure which respects these limitations, and nevertheless verifies quantum supremacy. 
 
The leading statistical measure proposed for verification is the ``cross-entropy'' \cite{boixo2016characterizing,boixo2017fourier,neill2017blueprint}, which is defined as
\[ \sum_x p^U_{dev}(x) \log \left( \frac{1}{p^U_{id} (x)} \right) \]
where $p^U_{dev}(x)$ is the probability the experimental device outputs $x$, and $p^U_{id}(x)$ is the probability the ideal device outputs $x$. This measure is specifically designed so that one can estimate it by taking a few samples $x_1,x_2,\ldots,x_k$ from the device, and computing the average value of $\log(1/p^U_{id}(x_i))$ using a classical supercomputer.

Ideally, we would like to connect the cross-entropy measure to the rigorous complexity-theoretic arguments in favor of quantum supremacy developed in Section \ref{sec:ourresults1}.  Invoking these hardness results as currently formulated requires the output distribution of the experimental quantum device to be close in total variation distance to the ideal.

Unfortunately, without strong assumptions as to how the quantum device operates, cross-entropy does not certify closeness in total variation distance -- in fact we give a counterexample distribution that achieves a nearly perfect cross-entropy score and yet is arbitrarily far from ideal in total variation distance. 

Another attempt at obtaining quantum supremacy from RCS is to make use of certain verifiable properties of the resulting ideal outcome distributions. Most notably, the Porter-Thomas ``shape'' of the RCS outcome distribution -- i.e., how many output strings $x$ have their output probability $p(x)$ in a certain range -- has been suggested as a ``signature'' of quantum effects \cite{boixo2016characterizing}.  We give an example of a naturally arising classical process that resembles the physics of a noisy/decoherent quantum system and yet has an outcome distribution that approximates Porter-Thomas. %
 
Consequently, any supremacy proposal based on outcome statistics cannot be based solely on shape. It must directly incorporate the relationship between specific outcome strings and their probabilities. Cross-entropy does take this into account because it requires computing the ideal output probabilities of the observed samples. It has been suggested that it may be intrinsically difficult to achieve high cross-entropy \cite{boixo2017fourier}, but this is thus far not supported by any complexity-theoretic evidence. Another recent proposal of Aaronson and Chen called Heavy Output Generation (or HOG) identifies a particularly simple relationship between output strings and their probabilities as a  possible avenue to supremacy \cite{aaronson2016complexity}. Viewed 
from the correct perspective, cross-entropy and HOG are more similar than they appear at first sight. While HOG can be tied to a hardness conjecture called QAUTH, a major challenge is to connect this with a standard, well-believed conjecture such as the non-collapse of the $\PH$.  Directly linking verification to computational complexity remains open for all supremacy proposals to date, including BosonSampling.

\section*{Acknowledgements}
We thank Scott Aaronson, Dorit Aharonov, Matthew Coudron, Abhinav Deshpande, Tom Gur, Zeph Landau, Nicholas Spooner, and Henry Yuen for helpful discussions. Authors AB, BF, CN, and UV were supported by ARO Grant W911NF-12-1-0541 and NSF Grant CCF-1410022 and a Vannevar Bush faculty fellowship.  Portions of this paper are a contribution of NIST, an agency of the US government, and is not subject to US copyright.
\bibliography{references}
\bibliographystyle{alpha}

\appendix

\section{Worst-to-average-case reduction}
\label{sec:worsttoavg}

Our main result is to give the first worst-to-average-case reduction for computing the output probabilities of random quantum circuits.
We will now describe why this result is critical to establishing quantum supremacy from Random Circuit Sampling (RCS).

Let us first define what we mean by RCS. Random Circuit Sampling is the process of picking a random (efficient) quantum circuit and then sampling from its output distribution.  Formally, an \emph{architecture} $\Aa$ is a collection of graphs, one for each integer $n$. Each graph consists of $m\leq\poly{(n)}$ vertices where each vertex $v$ has $\mathrm{deg}_{\mathrm{in}}(v) = \mathrm{deg}_{\mathrm{out}}(v) \in \{1,2\}$.  A circuit $C$ acting on $n$ qubits over $\Aa$ is instantiated by taking the $n$-th graph and specifying a gate for each vertex in the graph that acts on the qubits labelled by the edges adjacent to that vertex.  
That is, we can think of an architecture as an outline of a quantum circuit (one for each size n), and one needs to fill in the blanks (specify each gate) to instantiate a circuit.

We will consider the distribution on circuits where each gate is drawn uniformly at random. Here ``uniformly at random'' means according to the Haar measure, i.e. the unique measure on unitary matrices that is invariant under (left or right) multiplication by any unitary.
\begin{definition}[Haar random circuit distribution]
Let $\Aa$ be an architecture over circuits and let the gates in the architecture be $\{G_i\}_{i = 1, \ldots, m}$. Define the distribution $\Hh_\Aa$ (or $\Hh$ when $\Aa$ is clear from context) over circuits in $\Aa$ by drawing each gate $G_i$ independently from the Haar measure.
\end{definition}

Random Circuit Sampling is then defined as follows:

\begin{definition}[Random Circuit Sampling]
Random Circuit Sampling over a fixed architecture $\Aa$ is the following task: given a description of a random circuit $C$ from $\Hh_\Aa$, and a description of an error parameter $\epsilon>0$, sample from the probability distribution induced by $C$ (i.e., draw $y \in \bits^n$ with probability $\Pr(y) = \abs{\mel{y}{C}{0^n}}^2$) up to total variation distance $\epsilon$ in time $\poly(n,1/\epsilon)$. 
\end{definition}

While RCS is defined relative to an architecture $\Aa$, the exact choice of $\Aa$ will not matter for our main result, so we will often suppress this dependence in the next several sections. We will discuss the architectures proposed for quantum supremacy in detail in Appendix \ref{sec:anticonc}.
Also, note that this definition is designed to allow for a small amount of error in the classical sampler. This is to capture the fact that real-world quantum devices will be unable to perform this task exactly due to noise - and hence this definition allows the classical device the same error tolerance we allow the quantum device.
As usual \emph{total variation distance} means one half of the $\ell_1$ distance between the probability distributions.

The goal of our work is to argue that RCS is difficult for classical computers.
The crux of this argument lies in the relative difference in the difficulty of estimating the output probabilities of classical vs quantum circuits.
It is well known that certain output probabilities of quantum circuits are very difficult to compute -- in fact, they can be $\#\mathsf{P}$-hard to approximate, which is truly intractable.
In contrast, it is much easier to approximate the output probabilities of classical circuits \cite{stockmeyer85}, under an assumption known as the non-collapse of the polynomial hierarchy.
This result alone is enough to establish the difficulty of \emph{exactly} sampling from the probability distribution output by the quantum device (i.e. in the case $\epsilon=0$) \cite{BJS2010,aaronsonboson}.

However, to make this argument robust to experimental noise, we need the hardness of computing output probabilities to be ``more spread out'' in the output distribution, rather than concentrated in a single output which could be corrupted by noise. 
This was precisely the insight of Aaronson and Arkhipov \cite{aaronsonboson}. They showed that BosonSampling cannot be classically simulated under the following conjecture:

\begin{conjecture}[\cite{aaronsonboson}, Informal] \label{gpe-conjecture}
Approximating most output probabilities of most linear optical networks  is $\SHARPP$-hard.

\end{conjecture}

While they did not prove this conjecture, they were able to prove the following necessary worst-to-average-case reduction:
\begin{theorem}[\cite{aaronsonboson}, Informal]\label{gpe-theorem}
Exactly computing most output probabilities of most linear optical networks  is $\SHARPP$-hard.
\end{theorem}

Our Theorem \ref{averagecasethm} establishes the analogue of Theorem \ref{gpe-theorem} for Random Circuit Sampling.  Just as for Aaronson and Arkhipov, this theorem will give necessary evidence in support of our main hardness conjecture:

\begin{conjecture}[Informal] \label{missing-conjecture}
There exists an architecture $\mathcal{A}$ so that approximating $\abs{\bra{y}C\ket{0^n}}^2$ for most outcomes $y\in\{0,1\}^n$ and $C$ drawn from $\Hh_\Aa$ is $\SHARPP$-hard. 
\end{conjecture}

Furthermore, prior work has shown that Random Circuit Sampling has additional property known as anti-concentration \cite{brandao2013exponential,hangleiter2017anti}, which has not been proven for BosonSampling or FourierSampling. Anti-concentration can be seen as evidence that an average-case hardness result could be made robust to noise. 
We will discuss how known anti-concentration results can be integrated into our hardness proof in Appendix \ref{sec:anticonc}.

\subsection{Our average-case reduction}
\label{sec:average-case-reduction}

Our first result gives evidence that approximating average-case output probabilities of random quantum circuits remains difficult.  It is well known that computing output probabilities of worst-case quantum circuits is $\SHARPP$-hard.
Our goal is, therefore, to establish that computing output probabilities of \emph{average-case} random quantum circuits is just as difficult.
We achieve this by giving a \emph{worst-to-average-case reduction} for computing output probabilities of random quantum circuits. 
That is, we show that if one could compute average-case quantum circuit probabilities, then one could infer the value of worst-case quantum circuit probabilities.
Therefore, computing average-case probabilities is also $\SHARPP$-hard.

Establishing average-case hardness is surprisingly subtle. It will be useful to first recall the worst-to-average-case reduction for the permanent of matrices over the finite field $\FF_q$ \cite{lipton_1991}, where $q$ is taken to be a sufficiently large polynomial in the input parameter. In the case of permanents, the structure which connects the values of random permanents is low-degree polynomials. The permanent of an $n \times n$ matrix,
\[
\perm(A) = \sum_{\sigma \in S_n} \prod_{i = 1}^n A_{i,\sigma(i)}
\]
is a polynomial of degree $n$ in the $n^2$ matrix entries. Let $X$ be a random $n \times n$ matrix over $\FF_q$, where $q\geq n+2$.  Moreover, suppose our goal is compute the permanent of a worst-case matrix $Y$.  We first consider the line $A(t) = X t + Y$; note that for $t \neq 0$, $A(t)$ is uniformly distributed over $\FF_q^{n \times n}$. 
If we are able to calculate $\perm(R)$ with probability $\geq 1 - \frac{1}{3(n+1)}$ over $R \sim_{\mathcal{U}} \FF_q^{n \times n}$, then by the union bound, we could compute $A(t)$ correctly at $n+1$ different values of $t$ with bounded error probability. This is possible because the union bound holds despite $A(t)$ being correlated with one another -- it only requires that the \emph{marginal} distribution on each one is uniform. So standard polynomial interpolation techniques on $\{(t_j, \perm(A(t_j))\}_{j = 1, \ldots, n+1}$ allow us to learn the function $\perm(A(t))$ and therefore, $\perm(Y) = \perm(A(0))$. With more rigorous analysis -- but the same intuition -- one can argue that we only need to be calculate $\perm(R)$ with probability $3/4 + 1/\poly(n)$ \cite{welch1986error,Gemmell:1991:SPA:103418.103429}\footnote{Additional improvements have been made to reduce to the probability to $1/\poly(n)$ \cite{548475,Cai:1999:HP:1764891.1764903}.}. 

Therefore, polynomial interpolation allows us to compute permanents of every matrix $\in \FF_q^{n \times n}$ if we can compute the permanent on most matrices. 
A ``random survey'' of permanent values can be used to infer the value of all permanents. 
Combined with the fact that the permanent problem is worst-case $\SHARPP$-hard \cite{VALIANT1979189}, this implies that computing permanents in $\FF_q^{n \times n}$ on average is $\#\mathsf{P}$-hard. Formally, the following result was obtained.
\begin{theorem}[Average-case hardness for permanents \cite{lipton_1991,Gemmell:1991:SPA:103418.103429}]
The following is $\SHARPP$-hard: For sufficiently large $q$, given a uniformly random matrix $M \in \FF_q^{n \times n}$, output $\perm(M)$ with probability $\geq \frac{3}{4} + \frac{1}{\poly(n)}$.
\end{theorem}

To establish worst-to-average-case reductions for random circuits, we need to find a similar structural relation between our worst-case circuit, whose output probability we wish to compute, and average-case circuits in which each gate is chosen randomly.  A first observation is that there is indeed a low-degree polynomial structure -- stemming from the Feynman path-integral -- which allows us to write the probability of any outcome as a low-degree polynomial in the gate entries.  This polynomial is fixed once we fix both the outcome and the architecture of the circuit, and the degree is twice the number of gates in the circuit\footnote{The factor of 2 accounts for the Born rule for output probabilities.}, which is a polynomial in the input parameter.
\begin{fact}[Feynman path integral]
Let $C = C_mC_{m-1}\ldots C_2 C_1$, be a circuit formed by individual gates $C_i$ acting on $n$ qubits. Then 
\[\bra{y_m}C\ket{y_0} = \sum_{y_1,y_2, \ldots, y_{m-1} \in \bits^n} \prod_{j = 1}^m \bra{y_j}C_j\ket{y_{j-1}}.
\]
\label{fact-feynmanpathintegral}
\end{fact}

There are two subtleties we need to address.  The first is that the value of this polynomial is the probability of a fixed output $y_m$. Our analysis will therefore focus on the hardness of estimating the $\przero{C}\defeq|\langle 0^n|C|0^n\rangle|^2$ probability for $C$ drawn from $\Hh_\Aa$, rather than the hardness of approximating the probability of a random $y_m$. These can be proven equivalent using the ``hiding'' property of the  $\Hh_\Aa$ distribution: we can take a circuit drawn from this distribution, append Pauli $X$ gates to a uniformly chosen subset of output qubits, and remain distributed via $\Hh_\Aa$.  
We discuss hiding in more detail in Appendix \ref{subsec:samp-implies-count}.

The second subtlety is that this is a polynomial over the complex numbers, instead of $\mathbb{F}_q$. Bridging this gap requires considerable technical work\footnote{We note that Aaronson and Arkhipov have given a worst-to-average-case reduction for computing the permanent with complex Gaussian entries \cite{aaronsonboson}.  However, our reduction will be quite different, due to structural differences between quantum circuit amplitudes and permanents.}.  Indeed, in proving the reduction for permanents of matrices over finite fields, we leveraged the fact that $A(t) = Xt + Y$ will be randomly distributed across $\FF_q^{n \times n}$ since $X$ is uniformly random and $Y$ is fixed.  To leverage a similar property for random circuit sampling, we need, given a worst-case circuit $C$, a polynomial $C(t)$ over circuits such that (1) $C(0) = C$ and (2) $C(t)$ is distributed like a Haar random circuit.
To be more precise, for a fixed architecture $\Aa$, we will  we hope to say that the $\przero{C}$ probability for a circuit $C \in \Aa$ drawn from $\Hh_\Aa$ is hard to compute on average.

A na\"ive approach to doing this is to copy the proof for the permanent. If we could perturb each gate in a random linear direction, then we could use polynomial interpolation to perform the worst-to-average-case reduction as above. That is, consider taking a worst-case circuit $A$ and adding a random circuit $B$ (gate by gate) to obtain $A+tB$. It is true that $\przero{A + tB}$ is a low-degree polynomial in $t$, so one might hope to use interpolation to compute the worst-case value at $t=0$. Unfortunately, this idea does not work because quantum gates do not have a linear structure. In other words, if $A$ and $B$ are unitary matrices, then $A+tB$ is not necessarily unitary -- and hence $A+tB$ are not necessarily valid quantum circuits. So this na\"ive interpolation strategy will not work.

We consider a different way of perturbing circuits.  Suppose that we take a worst-case circuit $C = C_m, \ldots, C_1$, and multiply each gate $C_j$ by an independent Haar random matrix $H_j$. That is, we replace each gate $C_j$ with $C_jH_j$. By the left-invariance of the Haar measure, this is equivalent to selecting each gate uniformly at random. Now suppose we ``rotate back'' by tiny amount back towards $C_j$ by some small angle $\theta$. More specifically, replace each gate $C_j$ of the circuit with $C_j H_j e^{-i h_j \theta}$ where $h_j = - i \log H_j$. If $\theta=1$ this gives us back the worst-case circuit $C$, but if $\theta$ is very tiny this looks almost Haar random. %
One might hope that by collecting the values of many probabilities at small angles $\theta$, one could interpolate back to the worst-case point $C$ of interest. Therefore, a second attempt would be to take a (possibly worst-case) circuit $C$, scramble it by multiplying it gate-wise by a \emph{perturbed Haar distribution} defined below, and then use some form of interpolation in $\theta$ to recover the probability for $C$ at $\theta=1$.

\begin{definition}[$\theta$-\taylor]
Let $\Aa$ be an architecture over circuits, $\theta$ a constant $\in [0,1]$, and let $G_m, \ldots, G_1$ be the gate entries in the architecture. Define the distribution $\Hh_{\Aa, \theta}$ (or $\Hh_\theta$ when $\Aa$ is clear from context) over circuits in $\Aa$ by setting each gate entry $G_j = H_j e^{-i h_j \theta}$ where $H_j$ is an independent Haar random unitary and $h_j = - i \log H_j$.
\end{definition}

Unfortunately, this trick is not sufficient to enable the reduction. The problem is that $e^{-i \theta h_j}$ is not a low-degree polynomial in $\theta$, so we have no structure to apply polynomial interpolation onto. While there is structure, we cannot harness it for interpolation using currently known techniques.
Although this doesn't work, this trick has allowed us to make some progress.
A promising property of this method of scrambling is that it produces circuits which are close to randomly distributed -- which we will later prove rigorously.
This is analogous to the fact that $A+tB$ is randomly distributed in the permanent case, a key property used in that proof. 
We merely need to find some additional polynomial structure here in order to utilize this property.

We find this polynomial structure by considering Taylor approximations of $e^{-i h_j \theta}$ in place of $e^{-i h_j \theta}$ in the above construction. The values of the probabilities in this truncation will be very close to those of the proposal above -- and yet by construction we have ensured our output probabilities are low degree polynomials in $\theta$. Formally, we define a new distribution over circuits with this property:

\begin{definition}[$(\theta, K)$-\trunctaylor]
Let $\Aa$ be an architecture over circuits, $\theta$ a constant $\in [0,1]$, $K$ an integer, and let $G_m, \ldots, G_1$ be the gate entries in the architecture. Define the distribution $\Hh_{\Aa, \theta, K}$ (or $\Hh_{\theta,K}$ when $\Aa$ is clear from context) over circuits in $\Aa$ by setting each gate entry 
\[
G_j = H_j \left( \sum_{k = 0}^K \frac{(-i h_j \theta)^k}{k!} \right) 
\]
where $H_i$ is an independent Haar random unitary and $h_j = - i \log H_j$.
\end{definition}

Now suppose we take our circuit $C$ of interest and multiply each gate in it by $\Hh_{\theta,K}$ to ``scramble'' it. This is precisely how a computer would sample from $C\cdot \Hh_\theta$ as one cannot exactly represent a continuous quantity digitally. Suppose we could compute the probabilities of these circuits for many choices of $\theta$ with high probability. 
Now one can use similar polynomial interpolation ideas to show hardness of this task.

To state this formally, let us define some notation. For a circuit $C$ and $\Done$ a distribution over circuits of the same architecture, let $C \cdot \Done$ be the distribution over circuits generated by sampling a circuit $C' \sim \Done$ and outputting the circuit $C \cdot C'$ where multiplication is defined gate-wise.  
Explicitly, we show the following worst-to-average-case reduction, which we prove in Appendix \ref{subsec:proof-worst-to-avg}:

\begin{customthm}{\ref{averagecasethm}}
Let $\Aa$ be an architecture so that computing  $\przero{C}$ to within additive precision $2^{-\poly{(n)}}$, for any $C$ over $\mathcal A$ is  $\SHARPP$-hard in the worst case.  Then it is $\SHARPP$-hard to compute 8/9 of the probabilities $\przero{C'}$ over the choice of $C'$  from the distributions $\Dtwo_{C} \defeq C \cdot \Hh_{\theta,K}$ where $\theta = 1/\poly(n)$, $K = \poly(n)$.
\end{customthm}

\subsection{Theorem \ref{averagecasethm} is necessary for Conjecture \ref{missing-conjecture}}\label{subsection:necessity}
To reflect on our result, Theorem \ref{averagecasethm} shows that a worst-to-average-case reduction is indeed possible with respect to a distribution over circuits that is close to the Haar distribution we desire.  Of course, a skeptic could claim that such a result is only superficially related to our eventual goal, proving Conjecture \ref{missing-conjecture}.  Our next result is aimed precisely at such a skeptic: we show that the hardness result established in Theorem \ref{averagecasethm} will be necessary to prove Conjecture \ref{missing-conjecture}.

Let us start by choosing some convenient notation.  For the purposes of this section, let us fix an architecture $\mathcal{A}$ as well as parameters $\theta = \frac{1}{\poly(n)}$, and $K = \poly(n)$. 
Then, with respect to a fixed circuit $C$ over this architecture, we denote the distribution $C \cdot \Hh_{\theta}$ as $\Done_C$ (i.e., the corresponding $\theta-$\taylor), and $C \cdot \Hh_{\theta,K}$ will be denoted $\Dtwo_C$ (i.e., the corresponding $(\theta,K)-$\trunctaylor).
We also define the joint distribution of $\Done_C$ and $\Dtwo_C$, which we denote by $\Jj_C$.  This is the distribution over pairs of circuits $(C_1,C_2)$ generated by choosing independent Haar random gates $\{H_j\}_{j=1 \ldots m}$ and using this choice to publish $C_1$ from $\Done_C$ and $C_2$ from $\Dtwo_C$, using the same choice of $\{H_j\}$.  
Then, the marginal of $\Jj_C$ on $C_1$ is $\Done_C$ and on $C_2$ is $\Dtwo_C$ but they are correlated due to the same choice of $\{H_j\}$.
For simplicity of notation, we will often suppress the argument $C$ and simply write $\Done,\Dtwo,\Jj$.

Now we will show how to use the existence of an algorithm for computing probabilities of most circuits with respect to the \trunctaylor\ to estimate probabilities of most circuits drawn from the \haar.
We introduce one more helpful definition for these results, namely:
\begin{definition} We say an algorithm $\mathcal{O}$ $(\delta,\epsilon)$-computes a quantity $p(x)$ with respect to a distribution $F$ over inputs if:
\[\Pr_{x \sim F}\left[p(x)-\epsilon \leq \mathcal{O}(x) \leq p(x) + \epsilon\right]\geq 1-\delta.\]
\end{definition}
In other words, the algorithm computes an estimate to the desired quantity with high-probability over instances drawn from $F$.  In these terms, the main result of this section will be:
\begin{theorem}\label{thm:necessity}
Suppose there exists an efficient algorithm $\mathcal{O}$ that for architecture $\mathcal{A}$, $(\epsilon,\delta)$-computes the $\przero{C'}$ probability with respect to circuits $C'$ drawn from $\Dtwo$, then there exists an efficient algorithm $\mathcal{O'}$ that $(\epsilon',\delta')$-computes the $\przero{C'}$ probability with respect to circuits $C'$ drawn from $\Random$, with $\epsilon'=\epsilon+1/{\sf exp}(n)$ and $\delta'=\delta+1/\poly(n)$.
\end{theorem}

From this, one has the following immediate corollary:

\begin{cor}
Conjecture \ref{missing-conjecture} implies Theorem \ref{averagecasethm}.
\end{cor}
\begin{proof}
If there is an algorithm exactly computing probabilities over $\Dtwo$, then there is an algorithm approximately computing probabilities over $\Hh$. Therefore, if approximately computing probabilities over $\Hh$ is $\SHARPP$-hard, then exactly computing probabilities over $\Dtwo$ is $\SHARPP$-hard as well.
\end{proof}

In other words, our main result is necessary for the quantum supremacy conjecture (Conjecture \ref{missing-conjecture}) to be true.

We start proving Theorem \ref{thm:necessity} by establishing two facts which relate the distributions of circuits drawn from the joint distribution $\Jj$.  A natural interpretation of Facts \ref{fact:closeness-of-output} and \ref{fact:closeness-of-input} is as statements about the proximity of output probabilities and input distributions, respectively. Fact \ref{fact:closeness-of-output} states that the output probabilities of circuits drawn from the joint distribution $\Jj$ are effectively the same. Fact \ref{fact:closeness-of-input} states the perturbed distribution is essentially Haar -- therefore, choosing the inputs from the Haar distribution or the perturbed Haar distribution is immaterial.

\begin{fact}
Let $\Aa$ be an architecture over circuits and $C$ a circuit in the architecture. Let $(C_1,C_2)$ be circuits drawn from $\Jj$. Then the zero probabilities of $C_1$ and $C_2$ are close; namely,
\[
\abs{\przero{C_1} - \przero{C_2}} \leq 2^{-\poly(n)}.
\]
\label{fact:closeness-of-output}
\end{fact}

\begin{proof}
By expanding the exponential as a Taylor series, we can express each gate $C_{1,j}$ and $C_{2,j}$ of $C_1$ and $C_2$, respectively, as
\[
C_{1,j} = C_j H_j \left( \sum_{k = 0}^\infty \frac{(-ih_j\theta)^k}{k!} \right); \qquad 
C_{2,j} = C_j H_j \left( \sum_{k = 0}^K \frac{(-ih_j\theta)^k}{k!} \right).
\]
Therefore, $C_{1,j} - C_{2,j} = C_j H_j \left( \sum_{k = K+1}^{\infty} \frac{(-i h_j \theta)^k}{k!} \right)$. We can apply the standard bound on Taylor series to bound $\abs{\mel{y_j}{C_{1,j} - C_{2,j}}{y_{j-1}}} \leq \frac{\kappa}{K!}$ for some constant $\kappa$. Applying this to a Feynman path integral,
\[
\abs{\ev{C_1}{0} - \ev{C_2}{0}} \leq \sum_{y_1, \ldots, y_m} \abs{\prod_{j=1}^m \mel{y_j}{C_{1,j}}{y_{j-1}} - \prod_{j=1}^m \mel{y_j}{C_{2,j}}{y_{j-1}}   } \leq 2^{n(m-1)} \cdot \Oo\left( \frac{m\kappa}{K!} \right) = \frac{2^{O(nm)}}{(K!)^m}.
\]
This proves that the amplitudes are close. As the amplitudes have norm at most 1, then the probabilities are at least as close. The result follows by a sufficiently large choice of $K=\poly(n)$.
\end{proof}

\begin{fact}
Let $\Aa$ be an architecture on circuits with $m$ gates and $C \in \Aa$ a circuit from that architecture. Then the distribution $\Hh$ and $\Done$ are $O(1/\poly(n))$ close in total variation distance.
\label{fact:closeness-of-input}
\end{fact}

\begin{proof}

To prove this, we will show that for any particular gate of the circuit, the distributions induced by $\Hh$ and $\Done$ are $O(\theta)$ close in total variation distance. Then the additivity of total variation distance for independent events implies that the distributions are $O(m\theta)$-close (i.e. if $D$ and $D'$ are $\epsilon$-close in total variation distance, then $n$ independent copies of $D$ are $n\epsilon$-close to $n$ independent copies of $D'$). The result then follows from a suitably small choice of $\theta=1/\poly(n)$.

Now consider the distributions $\Hh$ and $\Done$ on a single two-qubit gate. 
Since the Haar measure is left-invariant, the distance between these is equivalent to the distance between $C\cdot \Hh$ and $\Done=C\cdot\Hh_{\theta}$.
Since total variation distance is invariant under left multiplication by a unitary, this is equivalent to the distance between $\Hh$ and $\Hh_{\theta}$.

Intuitively, the reason these are $O(\theta)$ close is as follows: consider a random rotation in $SO(3)$, vs. a random rotation in $SO(3)$ which has been ``pulled back'' towards the identity. By construction, the axes of rotations will be uniformly random over the sphere in both distributions. The only difference between the distributions lies in their angles of rotation -- the former's angle of rotation is uniform in $[0,2\pi]$ while the latter's is uniform in $[0,2\pi(1-\theta)]$. These distributions over angles are clearly $\theta$-close in total variation distance. This immediate implies these distributions over matrices are $\theta$-close in total variation distance as well since matrices are uniquely defined by the eigenbasis and eigenvalues.

We can extend this logic to the two-qubit case as well.
By construction the distributions $\Hh$ and $\Hh_{\theta}$ will be diagonal in a uniformly random basis $U$ (since ``pulling back'' a matrix $A$ by $e^{-i\theta \log A}$ preserves the eigenbasis).
Hence the only difference between these distributions lies in their distribution over eigenvalues.
We will show their distribution over eigenvalues are $O(\theta)$ close in total variation distance, which will imply the claim.
In particular, the distribution of eigenvalues $e^{i\theta_1},e^{i\theta_2},e^{i\theta_3},e^{i\theta_4}$ of a two qubit gate drawn from $\Hh$ is given by the density function, due to Weyl (e.g. \cite{diaconis1994eigenvalues}),
\[\Pr[\theta_i = \hat{\theta_i}] \propto \displaystyle\prod_{i\neq j} \left|e^{i\hat{\theta_i}} - e^{i\hat{\theta_j}}\right|^2.\]
In contrast the distribution over eigenvalues of a two-qubit gate drawn from $\Hh_{\theta}$ is
\[\Pr[\theta_i = \hat{\theta_i}] \propto \begin{cases} 
0 & \exists i:\hat{\theta_i} \geq 2\pi(1-\theta) \\
\displaystyle\prod_{i\neq j} \left|e^{i\hat{\theta_i}} - e^{i\hat{\theta_j}}\right|^2 & o.w.
\end{cases}
\]
One can easily compute that the total variation distance between these measures is $O(\theta)$, which implies the claim. This simply uses the fact that the above density function is smooth and Lipschitz, so a version of the same density function which has been ``shifted'' by $\theta$ is $O(\theta)$ close in total variation distance.
\end{proof}

Armed with these facts we are now ready to prove Theorem \ref{thm:necessity}.  We divide the proof into two steps, encapsulated into two lemmas (Lemmas \ref{Lemma:One}, \ref{Lemma:Two}).  In the first, we show how to use an algorithm that works on average over circuits drawn from $\Dtwo$ to get an algorithm that works on average over pairs of circuits drawn from $\Random$ and $\Done$.  

\begin{lemma}\label{Lemma:One}
Suppose there exists an algorithm $\mathcal{O}$ that for any circuit $C$ from a fixed architecture $\mathcal{A}$ takes as input a circuit $C_2$ sampled from $\Dtwo$ and $(\epsilon,\delta)$-computes the $\przero{C_2}$ probability.  Then there exists an algorithm $\mathcal{O}'$ that receives as input a Haar random circuit $C$ as well as a sample $C_1$ from $\Done$ and $(\epsilon',\delta)$-computes the $\przero{C_1}$ probability, where $\epsilon'=\epsilon+1/\exp(n)$.	
\end{lemma}
\begin{proof}
This lemma is primarily a consequence of Fact \ref{fact:closeness-of-output}.  Our objective in the proof will be to develop an algorithm $\mathcal{O}'$ that, given a circuit $C_1$ from the \taylor\ infers the corresponding circuit $C_2$ from the \trunctaylor.  Once it does this, it simply returns the output of $\mathcal{O}$ run on input $C_2$.  

More formally, consider an algorithm $\mathcal{O}'$ that is given as input $C$, as well as a pair of circuits $(C_1,C_2)\sim \Jj$, where $\Jj$ is the joint distribution with respect to $C$.  Then $\mathcal{O}'$ runs $\mathcal{O}$ on input $C_2$.  Clearly, from Fact \ref{fact:closeness-of-output}, the output probabilities of $C_1$ and $C_2$ are exponentially close, so we can see that $\mathcal{O}'$ $(\epsilon+1/\exp(n),\delta)$-computes the quantity $\przero{C_1}$.

Now by averaging over $C$, we see that in fact $\mathcal{O}'$ $(\epsilon+1/\exp(n),\delta)$-computes $\przero{C_1}$ with respect to a distribution over triplets of circuits $(C,C_1,C_2)$ in which $C$ is Haar random and the pair $(C_1,C_2)$ is distributed via the corresponding joint distribution $\Jj$.  Next notice that instead of receiving the triplet of inputs $(C,C_1,C_2)$, $\mathcal{O}'$ could simply have received a Haar random circuit $C$ and a circuit $C_1$ drawn from $\Done$.  This is because it can infer\footnote{This is simply because one can left-multiply $C_1$ by $C^\dagger$ to obtain the element drawn from $\Hh_{\theta}$. As $\theta$ is fixed beforehand, the algorithm can then deduce the corresponding element drawn from $\Hh$ with probability 1 by simply diagonalizing and stretching the eigenvalues by $1/(1-\theta)$. It can then compute the truncated Taylor series to obtain $C_2$.} the truncated circuit $C_2$ directly from $C$ and $C_1$. The Lemma follows.
\end{proof} \\

Next, we show how to use this new algorithm $\mathcal{O}'$ that works on average over pairs of circuits drawn from $\Random$ and $\Done$ to get an algorithm $\mathcal{O}''$ that works on average over circuits drawn from $\Random$.  
\begin{lemma}\label{Lemma:Two}
Suppose there exists an algorithm $\mathcal{O}'$ that takes as input a Haar random circuit $C$ from a fixed architecture $\mathcal{A}$ as well as a circuit $C_1$ drawn from $\Done$, and $(\epsilon,\delta)$-computes the $\przero{C_1}$ probability.  Then there exists an algorithm $\mathcal{O}'$ that $(\epsilon,\delta')$-computes the $\przero{C}$ probability with respect to input circuits $C$ drawn from $\Random$, with $\delta'=\delta+1/\poly(n)$.	
\end{lemma}
\begin{proof}
This lemma is a direct consequence of Fact \ref{fact:closeness-of-input}.  In particular, Fact \ref{fact:closeness-of-input} implies that the input distribution to the algorithm $\mathcal{O}'$, in which the first input circuit is Haar random and the second is drawn from $\Done$, is $1/\poly(n)$-close in total variation distance to the distribution over pairs of independently drawn Haar random circuits which we refer to as $\Hh^{(2)}$.  

Note total variation distance can be interpreted as the supremum over events of the difference in probabilities of those events.
Considering the event that $\mathcal{O}'$ is approximately correct in its computation of $\przero{C_1}$, this means if $\mathcal{O}'$ is run on inputs from the distribution $\Hh^{(2)}$ instead of from $C \sim \Hh$ and $C_1\sim\Done$, it will still be correct with high probability.
So $\mathcal{O}'$ will $(\epsilon,\delta+1/\poly(n))$-compute $\przero{C_1}$ with respect to this new distribution $\Hh^{(2)}$.  Now these two input circuits are independently drawn, and so $\mathcal{O'}$ can discard the unused input circuit.  We arrive at our Lemma.
\end{proof} \\

The results from Lemmas \ref{Lemma:One} and \ref{Lemma:Two} together prove Theorem \ref{thm:necessity}.

\subsection{Proof of Theorem \ref{averagecasethm}}
\label{subsec:proof-worst-to-avg}

\begin{customthm}{\ref{averagecasethm}}
Let $\Aa$ be an architecture so that computing  $\przero{C}$ to within additive precision $2^{-\poly{(n)}}$, for any $C$ over $\mathcal A$ is  $\SHARPP$-hard in the worst case.  Then it is $\SHARPP$-hard to compute 8/9 of the probabilities $\przero{C'}$ over the choice of $C'$  from the distributions $\Dtwo_{C} \defeq C \cdot \Hh_{\theta,K}$ where $\theta = 1/\poly(n)$, $K = \poly(n)$.
\end{customthm}

The proof of Theorem \ref{averagecasethm} follows by demonstrating the inherent \emph{polynomial} structure of the problem and leveraging the structure via polynomial interpolation to equate average-case and worst-case hardness. \\

\begin{proof}
Let $\{H_j\}$ be a collection of independent Haar random gates and define \[H_j'(\theta) = H_j \sum_{k = 0}^{K} \frac{(-i h_j \theta)^k}{k!}\] where $h_j = -i \log H_j$. Define the circuit $C'(\theta)$ as $C \cdot H'(\theta)$. Let $q(\theta) = \przero{C'(\theta)}$.

Notice that for a fixed choice of $\{H_j\}$, $q(\theta)$ is a low-degree polynomial in $\theta$. By a Feynman path integral (Fact \ref{fact-feynmanpathintegral}),
\begin{equation*}
\mel{y_m}{C'(\theta)}{y_0} = \sum_{y_1, \ldots, y_m \in \{0,\ldots, d-1\}^n} \prod_{j = 1}^m \mel{y_j}{[C'(\theta)]_j}{y_{j-1}}
\end{equation*}
is a polynomial of degree $mK$ as each term $\mel{y_j}{[C_1(\theta)]_j}{y_{j-1}}$ is a polynomial of degree $K$. Therefore, $q$ is a polynomial of degree $2mK$. Now assume that there exists a machine $\Oo$ such that $\Oo$ can compute $\przero{C'}$ for $8/9$ of $C'$ where $C'$ is drawn from the distribution in the theorem statement. A simple counting argument shows that for at least $2/3$ of the choices of $\{H_j\}$, $\Oo$ correctly computes $\przero{C'(\theta)}$ for at least $2/3$ of $\theta$. Call such a choice of $\{H_j\}$ good.

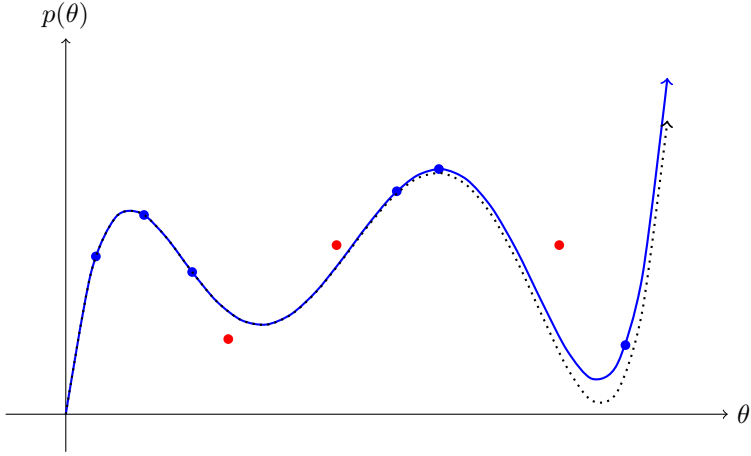
\begin{figure}[!ht]
\centering
\begin{tikzpicture}[xscale=8.0,yscale=5.0]
      \draw[->] (-0.1,0) -- (1.1,0) node[right] {$\theta$};
      \draw[->] (0,-0.1) -- (0,1.0) node[above] {$p(\theta)$};
      \draw[domain=0:1,smooth,variable=\x,blue,->,thick] plot ({\x},{12.4968 * \x - 94.6583 * \x * \x + 268.974 * \x * \x * \x - 316.667 * \x * \x * \x * \x + 130.754 * \x * \x * \x * \x * \x});
      \draw (0.13,0.530206) node[circle,fill,inner sep=1.3pt, blue] {};
      \draw (0.55,0.593288) node[circle,fill,inner sep=1.3pt, blue] {};
      \draw (0.21,0.378409) node[circle,fill,inner sep=1.3pt, blue] {};
      \draw (0.93,0.183831) node[circle,fill,inner sep=1.3pt, blue] {};
      \draw (0.05,0.419878) node[circle,fill,inner sep=1.3pt, blue] {};
      \draw (0.62,0.652424) node[circle,fill,inner sep=1.3pt, blue] {};
      
      \draw (0.82,0.45) node[circle,fill,inner sep=1.3pt, red] {};
      \draw (0.27,0.20) node[circle,fill,inner sep=1.3pt, red] {};
      \draw (0.45,0.45) node[circle,fill,inner sep=1.3pt, red] {};
      
      \draw[domain=0:1,smooth,variable=\x,black, thick, dotted, ->] plot ({\x},{12.4968 * \x - 94.6583 * \x * \x + 268.974 * \x * \x * \x - 316.667 * \x * \x * \x * \x + 130.64 * \x * \x * \x * \x * \x});
      
    \end{tikzpicture}
\caption{Example of a true function $\przero{C}$ (dotted), inherent polynomial $q(\theta) = \przero{C'(\theta)}$ (solid), and potentially noisy samples $\{(\theta_\ell, \Oo(\theta_\ell))\}$.}
\end{figure}

Consider a machine $\Oo'$ with fixed $\theta_1, \ldots, \theta_k \in [0,\frac{1}{\poly(n)})$ that queries $\Oo(\theta_\ell)$ for $\ell = 1, \ldots, k$. Then $\Oo'$ applies the Berlekamp-Welch Algorithm \cite{welch1986error} to compute a degree $2mK$ polynomial $\tilde{q}$ from the points $\{ (\theta_\ell, \Oo(\theta_\ell))\}_{\ell = 1, \ldots, k}$ and returns the output $\tilde{q}(1)$.

\begin{theorem}[Berlekamp-Welch Algorithm \cite{welch1986error}]
Let $q$ be a degree $d$ univariate polynomial over any field $\FF$. Suppose we are given $k$ pairs of $\FF$ elements $\{(x_i,y_i)\}_{i = 1, \ldots, k}$ with all $x_i$ distinct with the promise that $y_i = q(x_i)$ for at least $\min(d + 1, (k + d)/2)$ points. Then, one can recover $q$ exactly in $\poly(k,d)$ deterministic time.
\end{theorem}

We remark that if we choose $k = 100mK$, then for a good $\{H_j\}$ with high probability (by a Markov's inequality argument), the polynomial $\tilde{q} = q$. Therefore, $\tilde{q}(1) = q(1) = \przero{C'(1)}$. Since at least $2/3$ of $\{H_j\}$ are good, by repeating this procedure $O(1)$ times and applying a majority argument, we can compute $\przero{C'(1)}$ exactly. It only remains to show that $\przero{C'(1)}$ is a $2^{-\poly(n)}$ additive approximation to $\przero{C}$, a $\SHARPP$-hard quantity.

We can apply Fact \ref{fact:closeness-of-output} to argue that $\abs{\przero{C'(1)} - \przero{C}}$ is at most $2^{O(nm)}/{((K)!)^m}$. As we choose $K = \poly(n)$, this is at most $2^{-\poly(n)}$ for every desired polynomial.
\end{proof}

\subsection{Sampling implies average-case approximations in the polynomial hierarchy}
\label{subsec:samp-implies-count}

In this section, we explain why Conjecture \ref{missing-conjecture} implies quantum supremacy for RCS.
In particular, we show that such an efficient classical algorithm for RCS would have surprising complexity consequences. This section will be very similar to analogous results in earlier work (see e.g., \cite{aaronsonboson,feffermanfourier,bremner2016average}).

That is, we show that the following algorithm which we call an approximate sampler, is unlikely to exist:
\begin{definition}[Approximate sampler]
An approximate sampler is a classical probabilistic polynomial-time algorithm that takes as input a description of a quantum circuit $C$, as well as a parameter $\epsilon$ (specified in unary) and outputs a sample from a distribution $D_C'$ such that \[||D_C-D_C'||\leq \epsilon\] where $D_C$ is the outcome distribution of the circuit $C$ and the norm is total variation distance.
\end{definition}

Our main result will connect the existence of an approximate sampler to an algorithm which will estimate the probabilities of most Haar random circuits, in the following sense:

\begin{definition}[Average-case approximate solution] A polynomial-time algorithm $\mathcal{O}$ is an average-case approximate solution to a quantity $p(x)$ with respect to an input distribution $\mathcal{D}$ if:

 \[\Pr_{x\sim \Dd} \left[ \abs{\mathcal{O}(1^{1/\epsilon},1^{1/\delta},x)-p(x)} \leq \frac{\eps}{2^n}  \right] \geq 1 - \delta.\]

\end{definition}
In other words, an average-case approximate solution outputs a good estimate to the desired quantity for most random inputs but might fail to produce any such estimate for the remaining inputs.

More formally, the main theorem of this section, Theorem \ref{thm:samp-implies-count}, proves that the existence of an approximate sampler implies the existence of an average-case approximate solution for computing the $\przero{C}$ probability of a random circuit $C$ drawn from the Haar distribution.  This average-case approximate solution will run in probabilistic polynomial time with access to an $\NP$ oracle.  The main theoretical challenge in quantum supremacy is to give evidence that such an algorithm does not exist.  This would certainly be the case if the problem was $\SHARPP$-hard, or as hard as counting the number of solutions to a boolean formula.  Such a conjecture lies at the heart of all current supremacy proposals. More formally, this conjecture is:

\begin{customconj}{\ref{missing-conjecture}}
There exists a fixed architecture $\mathcal{A}$ so that computing an average-case approximate solution to $\przero{C}$ with respect to $\Hh_\Aa$ is $\SHARPP$-hard. 
\end{customconj}

We now show how Conjecture \ref{missing-conjecture} would rule out a classical approximate sampler for RCS, under well-believed assumptions. Specifically, assuming this conjecture is true, Theorem \ref{thm:samp-implies-count} tells us that an approximate sampler would give an algorithm for solving a $\SHARPP$-hard problem in $\BPP^\NP$.  Now, $\BPP^\NP$ is known to be in the third-level of the $\PH$ (see e.g., \cite{lautemann1983}).  In other words, $\BPP^\NP \subseteq \Sigma_3$.  On the other hand, a famous theorem of Toda tells us that all problems solvable in the $\PH$ can be solved with the ability to solve $\SHARPP$-hard problems.  That is, $\PH\subseteq\Poly^{\SHARPP}$ \cite{Toda}.  
Putting everything together, we have that an approximate sampler would imply that $\PH\subseteq\Sigma_3$, a collapse of the $\PH$ to the third-level, a statement that is widely conjectured to be false (e.g., \cite{Karp:1980:CNU:800141.804678,Boppana:1987:CSI:31193.31202}).

Finally, we prove Theorem \ref{thm:samp-implies-count}. The proof utilizes a classic theorem by Stockmeyer \cite{stockmeyer85}, which we state here for convenience.

 \begin{theorem}[Stockmeyer \cite{stockmeyer85}]
 Given as input a function $\bitsfn f n m$ and $y \in \bits^m$ there is a procedure that runs in randomized time $\poly(n,1/\eps)$ with access to a $\NP^f$ oracle that outputs an $\alpha$ such that
 \begin{equation*}
 (1 - \eps) p \leq \alpha \leq (1 + \eps) p \text{\ for \ } p = \Pr_{x \sim \Uu(\bits^n)} [f(x) = y].
 \end{equation*}
 \label{theorem-stockmeyer}
 \end{theorem}
In the context of this work, the primary consequence of Stockmeyer's theorem is that we can use an $\NP$ oracle to get a multiplicative estimate to the probability of any outcome of an approximate sampler, by counting the fraction of random strings that map to this outcome.   Using this idea we prove:
\begin{theorem}
If there exists an approximate sampler $\mathcal{S}$ with respect to circuits from a fixed architecture $\Aa$, there also exists an average-case approximate solution in $\BPP^{{\NP}^{\mathcal{S}}}$ for computing the $\przero{C}$ probability for a random circuit $C$ drawn from $\Hh_\Aa$.
\label{thm:samp-implies-count}
\end{theorem}

\begin{proof}
We start by proving a related statement, which says that if we can sample approximately from the outcome distribution of any quantum circuit, we can approximate most of the output probabilities of all circuits $C$.  This statement, unlike the  Theorem \ref{thm:samp-implies-count}, is architecture-agnostic.  
\begin{lemma}\label{lemma:nohiding}
	If there exists an approximate sampler $\mathcal{S}$ then for any quantum circuit $C$, there exists an average-case approximate solution in $\BPP^{{\NP}^{\mathcal{S}}}$ for computing the $|\bra{y}C\ket{0}|^2$ probability of a randomly chosen outcome $y \in \bits^n$.
\end{lemma}
\begin{proof}
First fix parameters $\delta, \eps > 0$. Then for any quantum circuit $C$, $\mathcal{S}(C, 1^{1/\eta})$ samples from a distribution $\eta$-close to the output distribution $p$ of $C$.  We denote this approximate outcome distribution by $q$. By Theorem \ref{theorem-stockmeyer}, there exists an algorithm $\Oo \in \BPP^{\NP^\mathcal{S}}$ such that
\begin{equation*}
(1- \gamma) q_y \leq \abs{ \Oo'(C,y,1^{1/\eps},1^{1/\gamma}) } \leq (1 + \gamma) q_y.
\end{equation*}
Let $\tilde q_y = \Oo(C,y,1^{1/\eta},1^{1/\gamma})$ for $\gamma$ to be set later. Since $q$ is a probability distribution, $\mathbb{E}(q_y) = 2^{-n}$. By Markov's inequality,
\begin{equation*}
\Pr_y \left[ q_y \geq \frac{k_1}{2^n} \right] \leq \frac{1}{k_1}; \qquad \Pr_y \left[ \abs{q_y - \tilde {q_y}} \geq \frac{\gamma k_1}{2^n} \right] \leq \frac{1}{k_1}.
\end{equation*}
Secondly, let $\Delta_y = \abs{p_y - q_y}$. By assumption, $\sum_y \Delta_y = 2 \eta$ so, therefore, $\mathbb{E}(\Delta_y) = 2 \eta / 2^n$. Another Markov's inequality gives
\begin{equation*}
\Pr_y \left[\Delta_y \geq \frac{2 k_2 \eta}{2^n} \right] \leq \frac{1}{k_2}.
\end{equation*}
With a union bound and a triangle inequality argument,
\begin{equation*}
\Pr_y \left[ \abs{p_y - \tilde q_y} \geq \frac{\gamma k_1 + 2 k_2 \eta}{2^n} \right] \leq \frac{1}{k_1} + \frac{1}{k_2}
\end{equation*}
Choose $k_1 = k_2 = 2/\delta, \gamma = (\eps \delta)/4, \eta = \gamma / 2$. Then,
\begin{equation*}
\Pr_y \left[ \abs{p_y - \tilde q_y} \geq \frac{\eps}{2^n} \right] \leq \delta.
\end{equation*}
Therefore, for any circuit $C$, the algorithm $\mathcal{O}$ is an approximate average-case solution with respect to the uniform distribution over outcomes, as desired.
\end{proof}

Now we use the shared architecture constraint in the theorem statement to enable a so-called \emph{hiding} argument. Hiding shows that if one can approximate the $\abs{\mel{y}{C}{0}}^2$ probability for a random $y \in \bits^n$, the one can also approximate $\przero{C}$ for a random $C$. This latter step will be crucial to our main result.  In particular, both the anti-concentration property and our proof of average-case hardness of estimating circuit probabilities relies on considering a fixed output probability (see Appendix \ref{subsec:proof-worst-to-avg} and \ref{sec:anticonc}).  

To prove this, we rely on a specific property of $\Hh_\Aa$.  This hiding property is that for any $C \sim \Hh_\Aa$, and uniformly random $y \in \bits^{n}$, $C_y \sim \Hh_\Aa$ where $C_y$ is the circuit such that $\bra{z}C_y\ket{0} = \bra{z\oplus y}C\ket{0}$.  In other words, the distribution over circuits needs to closed under appending Pauli $X$ gates to a random subset of output qubits. 

Lemma \ref{lemma:nohiding} tells us that for any circuit $C$, an approximate sampler gives us the ability to estimate most output probabilities $\bra{y}C\ket{0}$.  If we instead restrict ourselves to Haar random circuits over the architecture $\mathcal{A}$, we can think of this same algorithm $\mathcal{O}$ as giving an average-case approximate solution with respect to the distribution generated by first choosing $C$ from the Haar distribution and then appending $X$ gates to a uniformly chosen subset of the output qubits, specified by a string $y\in\{0,1\}^n$, since $\bra{y}C\ket{0}=\bra{0}C_y\ket{0}$.  Using the hiding property this is equivalent to an average-case approximate solution with respect to circuits $C$ drawn from the Haar random distribution over $\mathcal{A}$, as stated in Theorem \ref{thm:samp-implies-count}.
\end{proof}

\subsection{Connecting with worst-case hardness and anti-concentration}
\label{sec:anticonc}

Prior to this subsection, all of our results have been architecture agnostic-- our worst-to-average case reduction in Appendix \ref{subsec:proof-worst-to-avg} aims to reduce the presumed worst-case hardness of computing output probabilities of quantum circuits over a fixed architecture $\mathcal{A}$ to computing them on average over $\Hh_\Aa$.

Of course, for these results to be relevant to quantum supremacy, we need to establish that for the architectures $\Aa$ used in supremacy experiments, computing worst-case output probabilities is $\SHARPP$-hard.
Then our worst-to-average-case reduction shows that computing average case probabilities for these experiments over $\Hh_\Aa$ is $\#\mathsf{P}$-hard -- which is precisely what is necessary for the supremacy arguments of Appendix \ref{subsec:proof-worst-to-avg} to hold.
In this section, we will show that this requirement on $\Aa$ is quite mild.
In particular, we will show that a candidate instantiation of RCS which is known to anti-concentrate -- namely random quantum circuits on a 2D grid of depth $O(n)$ -- easily satisfy this property.
Therefore it is possible to have a single candidate RCS experiment which has both average-case $\#\mathsf{P}$-hardness as well as anti-concentration.

Such worst-case hardness can be established via the   arguments of Bremner, Jozsa and Shepherd \cite{BJS2010}.
Although we will not summarize these standard arguments here, the key technical ingredient is demonstrating that quantum computations over this fixed architecture are universal. This will imply that the power of the corresponding complexity class supplemented with the ability to do post-selected measurements is equal in power to $\PostBQP=\PP$ by a result of Aaronson \cite{Aar05}. That is, to show our worst-case hardness result it suffices to  show that the class of problems solvable by circuits over a fixed architecture is equal to $\BQP$.  
This can be established by standard results from measurement-based quantum computation involving universal resource states \cite{raussendorf2001one,raussendorf2003measurement,briegel2009measurement}.
Roughly speaking, these results allow us to prepare a fixed state on a 2D grid and simulate any quantum circuit by performing a sequence of adaptive one-qubit measurements on this state.
Combining these results immediately implies that if an architecture $\Aa$ is capable of generating one of these universal resource states, then $\Aa$ contains $\#\mathsf{P}$-hard instances -- because one could simply post-select the measurement outcomes such that no adaptivity is required. 

To be more formal, let us  define some notation. Let $\Aa\subseteq \Aa'$ if the gates in $\Aa$ are a subset of those in $\Aa'$.
Then if a a circuit $C$ is realizable in $\Aa$, then it is also realizable in $\Aa'$ - simply by setting those gates not in $\Aa$ to the identity\footnote{One can also expand this definition to consider a one-qubit gate to be a subset of a two-qubit gate - as one can always set the two-qubit gate to be the identity tensor a one qubit gate.}.
Consider  the ``brickwork'' state defined by Broadbent, Fitzsimons and Kashefi \cite{broadbent2009universal}.
The brickwork state $\ket{\Psi_{\mathrm{brick}}}$ is a universal resource state for measurement-based quantum computation, which has nice properties. In particular it can be prepared by a constant-depth quantum circuit $C_{\mathrm{brick}}$ on a 2D grid, where gates only act on nearest-neighbor qubits.
Let $\Aa_{\mathrm{brick}}$ be the architecture of $C_{\mathrm{brick}}$, adding on space for one-qubit gates on every output qubit.
Then $\Aa_{\mathrm{brick}}$ is universal for quantum computation under post-selection by the above arguments. 
Therefore these prior results immediately yield the following Lemma:
\begin{lemma}
\label{lem:brickwork}
For any architecture $\Aa$ such that $\Aa_{\mathrm{brick}} \subseteq \Aa$, it is $\#\mathsf{P}$-hard to compute worst case probabilities in $\Aa$.
\end{lemma}

Note that the condition required to invoke Lemma \ref{lem:brickwork} is extremely mild. It simply says that the architecture must contain a simple constant-depth nearest-neighbor circuit on a 2D grid as a subgraph. 
We now show that the mildness of this condition allow us to easily connect worst-case hardness to anti-concentration.

Let us first define anti-concentration and state why it is important in the context of quantum supremacy.
Broadly speaking, anti-concentration is a statement about the distribution of probabilities. It states that \emph{most} output probabilities are reasonably large. 
\begin{definition}[Anti-concentration]
For a fixed architecture $\Aa$, we say that RCS anti-concentrates on $\Aa$, if there exists constants $\kappa,\gamma > 0$ so that:
\[
\Pr_{C \sim \Hh_\Aa} \left[ \przero{C} \geq \frac{1}{\kappa 2^n} \right] \geq 1-\gamma.
\]
\end{definition}

Crucially, this anti-concentration property allows us to reduce the hardness of average-case approximate solutions (which, by definition, approximate the desired circuit probability \emph{additively}) to an average-case solution that approximates the solution \emph{multiplicatively}.  As such, we can at least ensure that these approximations are non-trivial, that is the signal is not lost to the noise.  More formally, 
\begin{lemma}
For a fixed architecture $\Aa$ for which RCS anti-concentrates, if there exists an algorithm $\Oo$ that estimates $\przero{C}$ to additive error $\pm \epsilon/2^n$ for a $1-\delta$ fraction of $C \sim \Hh_\Aa$, then $\Oo'$ also can be used to estimate $\przero{C}$ to multiplicative error $\eps \cdot \kappa$ for a $1 - \delta - \gamma$ fraction of $C \sim \Hh_\Aa$.
\end{lemma}

\begin{proof}
A rephrasing of the additive error assumption is $ \Pr_{C \in \Hh} \left[ \abs{\Oo(C) - \przero{C}} > \frac{\eps}{2^n} \right] \leq \delta$. We apply a union bound to argue that
\begin{align*}
\Pr_{C \in \Hh} \left[ \abs{\Oo(C) - \przero{C}} > \eps\kappa \przero{C} \right] &\leq \Pr_{C \in \Hh} \left[ \abs{\Oo(C) - \przero{C}} > \frac{\eps}{2^n} \right] + \Pr_{C \in \Hh} \left[ \frac{\eps}{2^n} > \eps \kappa \przero{C} \right] \\
&\leq \delta + \gamma.
\end{align*}
\end{proof}

Anti-concentration is known for random quantum circuits of depth $O(n)$.
It is possible to show that this instantiation of RCS obeys the conditions of Lemma \ref{lem:brickwork}, and hence can exhibit both average-case hardness and anti-concentration simultaneously.
More specifically, suppose that at each step one picks a random pair of nearest-neighbor qubits on a line, and applies a Haar random gate between those qubits, until the total depth of the circuit is $O(n)$.
Prior work has established that such circuits are approximate quantum two-designs, i.e. they approximate the first two moments of the Haar measure \cite{brandao2013exponential,brandao2016local}.
This, combined with the fact that unitary two-designs are known to anti-concentrate (which was noted independently in multiple works \cite{hangleiter2017anti,boulandcccs,bremner2017upcoming}), implies that random circuits of depth $O(n)$ anti-concentrate.
These results immediately generalize to random circuits of depth $O(n)$ on a 2D grid.
Note one can easily show that with probability $1-o(1/\poly(n))$ over the choice of a random circuit in this model, the architecture of the circuit obeys Lemma \ref{lem:brickwork}.
Hence, computing average-case probabilities over this random circuit model is $\#\mathsf{P}$-hard\footnote{Although here we are discussing average-case hardness over a random choice of architecture, this result easily follows from our reduction for a single architecture, since w.h.p. the architecture drawn is hard on average.}.
Therefore, random circuits of depth $O(n)$ on a 2D grid obtain both average-case hardness and anti-concentration.
We note that it is conjectured that random circuits of depth $O(n^{1/2})$ on a 2D grid anti-concentrate as well \cite{boixo2016characterizing}. If this conjecture is true then such circuits would also exhibit both anti-concentration and average-case hardness, as we only require constant depth to satisfy Lemma \ref{lem:brickwork}.

\section{Verification of Random Circuit Sampling}

\subsection{Technical preliminaries}

In this section, if unspecified, a probability distribution will be over strings $x\in\{0,1\}^n$. The size of the domain will be denoted $N=2^n$. The phrase ``with high probability'' will mean with probability $1-o(1)$.

\begin{definition}
Given two probability distributions $D,D'$, the cross-entropy, cross-entropy difference, and total variation distance between $D$ and $D'$, denoted $\ce(D,D')$, $\ced(D,D')$, and $|D-D'|$, respectively, are given by
\begin{align*}
\ce(D,D') &= \displaystyle\sum_{x\in\{0,1\}^n} D(x)\log\left(\frac{1}{D'(x)}\right), \\
\ced(D,D') &= \displaystyle\sum_{x\in\{0,1\}^n} \left(\frac{1}{N} -D(x)\right)\log\left(\frac{1}{D'(x)}\right), \\
|D-D'| &= \frac{1}{2}\displaystyle\sum_{x\in\{0,1\}^n} \left|D(x)-D'(x)\right|.
\end{align*}
\end{definition}
The cross-entropy difference is simply equal to $\ce(\mathcal{U},D') - \ce(D,D')$, where $\mathcal{U}$ is the uniform distribution.
One particular probability distribution which will play an important role in this discussion is the Porter-Thomas distribution. It approximately describes the probability distributions output by Haar random quantum circuits (see e.g., \cite{porterthomas,boixo2016characterizing}).

\begin{definition}
The Porter-Thomas distribution, $PT$, is the probability density function over $[0,\infty)$ defined as
\[f_{PT}(q) = Ne^{-qN}.\]
\label{def:pt}
\end{definition}

Let $\ket{\Psi_U} = U\ket{0^n}$ be the state obtained by applying the unitary $U$ to the all 0's input state. Let $p_U(x)$ denote the probability of obtaining $x$ upon measuring $\ket{\Psi_U}$, i.e. 
\[p_U(x) = \left|\braket{x}{\Psi_U}\right|^2.\]
Then, we have that for any $x$ the distribution of $p_U(x)$ over Haar random $U$ is well-approximated by the Porter-Thomas distribution.  For fixed outcome $x$, we will call this distribution over the choice of Haar random $U$, $P(x)$.

\begin{fact}[Thm. 35 of \cite{aaronsonboson}]
  For any fixed outcome $x$ and $c>0$, $|P(x)-PT| \leq O(1/N^c)$.
\end{fact}

We will also be interested in the joint distribution generated by choosing a Haar random unitary $U$ and considering the output probabilities of $k$ fixed outcomes $x_1,...,x_k$.  We will denote this distribution over vectors of length $k$ as $P(x_1,...,x_k)$. Not surprisingly, this can be approximated by $k$ i.i.d copies of the Porter-Thomas distribution\footnote{For instance, Aaronson and Arkhipov showed that, for $k=O(N^{1/6})$, $P(x_1,...,x_k)$ is approximately $PT^{(k)}$, up to small total variation distance error \cite{aaronsonboson}.}, $PT^{(k)}$.
For convenience, we will define $P=P(x_1,x_2,...,x_N)$.  

Although $P$ is not close in total variation distance to $PT^{(N)}$ \footnote{This is because any $v$ drawn from $P$ will satisfy $|v|_1=1$, while $PT^{(N)}$ will satisfy this condition with probability 0.}, the distribution $P$ does maintain some of the coarse-grained features of $PT^{(N)}$.
This is because an equivalent way of sampling from $P$ is to a draw a vector from $PT^{(N)}$ and renormalize so that  $|v|_1=1$ \cite{ozolsrandom}.
By concentration of measure, this renormalization factor will be close to 1 with very high probability.
Therefore, following \cite{boixo2016characterizing}, in this section we will often perform calculations using the heuristic of replacing $P$ with $PT^{(N)}$.
We will describe why this suffices for the calculations in which it is used.

\subsection{The cross-entropy supremacy proposal}

Cross-entropy is a leading proposal for verifying quantum supremacy \cite{boixo2016characterizing,boixo2017fourier,neill2017blueprint}. For RCS it provides a measure of the distance between 
the output distribution of the experimental device and the ideal random circuit $U$. 
Estimating it requires just taking $k \ll N$ samples, $x_1, \ldots, x_k$, from the experimental device, followed by the computation of the empirical estimate $E$ of the cross-entropy
\begin{equation}
E = \frac{1}{k}\displaystyle\sum_{i=1\ldots k} \log\left(\frac{1}{p_U(x_i)}\right)
\label{eq:measuring-ce}
\end{equation}
by using a supercomputer to calculate ideal probabilities $p_U(x_i)$ for only the observed outcome strings $x_i$. By the law of large numbers, after sufficiently many samples\footnote{An argument is made that taking $\sim 10^5$ samples for an $n = 50$ qubit device suffices to obtain a good estimate of $\ced(p_{dev},p_U)$ \cite{martinistalk}. Furthermore, it is argued that drawing the unitary $U$ from an approximate 2-design instead of drawing $U$ Haar randomly is sufficient to argue that the ideal device obeys $\ced(p_U,p_U) \approx 1$. We note that since $\log$ is not a low-degree polynomial, this is not guaranteed by the fact that the distribution drawn from is an approximate 2-design. The argument is made on the basis of numerical evidence \cite{boixo2016characterizing}.} $E$ will converge to $\ce(p_{dev}, p_U)$, where $p_{dev}$ is the output probability of their experimental device tuned to perform $U$. %
Since $\ced(p_{dev},p_U) = \ce(\mathcal{U},p_U)-\ce(p_{dev},p_U)$ and $\ced(\mathcal{U}, p_U)$ is closely concentrated about its mean, $ \log N + \gamma$, where $\gamma$ is Euler's constant, from this one can infer an estimate of $\ced(p_{dev},p_U)$.
The goal of their experiment is to achieve a value of $\ced(p_{dev},p_U)$ as close to the ideal expectation value as possible (with high probability).
In fact, this measure has become incredibly important to the Google/UCSB group: it is being used to calibrate their candidate experimental device \cite{neill2017blueprint,martinistalk}.

If the experimental device were actually identical to the ideal $U$, then the expected value of cross-entropy difference for Haar random $U$ is easily estimated:
\[
\mathbb{E}_{U \sim \mathcal{H}}\left[\ced(p_U, p_U)\right] =1\pm o(1).
\]
This follows from linearity of expectation, since one only needs to compute this quantity for individual $x \in \bits^n$, which approximately obey the Porter-Thomas distribution, one can compute this with a simple integral. However, any near-term implementation of this experiment will be subject to experimental noise and, therefore, one should not expect achieve exactly $\ced = 1$.
In fact, the Google/UCSB group expects to obtain $\ced>0.1$ on their 40-50 qubit device \cite{boixo2016characterizing}. Clearly, achieving a value of $\ced$ close to $1$ is necessary for their device to be functioning properly. Here we ask if it is sufficient as well, i.e. whether or not achieving $\ced = 1\pm\epsilon$ certifies that the device has achieved quantum supremacy. %

\subsection{Cross-entropy does not verify total variation distance}
\label{sec:cross-entropy-neq-tv}

Our results from Appendix \ref{sec:worsttoavg} provide evidence that it is classically hard to sample from any outcome distribution close in total variation distance to the ideal distribution.  %
Consequently, our goal in this section is to examine if achieving a sufficiently high cross-entropy difference can be used to certify that the observed outcome distribution is close to ideal in total variation distance. 

That is, we ask if for general distributions $D$, does achieving $\ced(D,p_U) = 1 - \epsilon$ for Haar typical $U$ certify that $|D-p_U|<f(\epsilon)$ for some function $f$ of $\epsilon$? This is not a priori impossible; for instance, Pinsker's inequality states that the square root of the KL divergence between two distributions, which is closely related to cross entropy, is an upper bound on the total variation distance. So in some sense, we are asking if cross-entropy behaves similarly to KL divergence in this manner.

We answer this question in the negative. Therefore, achieving high cross-entropy difference does not allow us to conclude quantum supremacy based on the results in Appendix \ref{sec:worsttoavg}.
\begin{theorem}
For every unitary $U$, there exists a distribution $D_U$ so that, with probability $1-o(1)$ over the choice of $U$ from the Haar measure, $|D_U-p_U| \geq 0.99$, and yet $\ced(D_U,p_U)  \geq 1 - O(1/N^{\Theta(1)})$, i.e. the cross-entropy difference is exponentially close to its ideal value. 
\end{theorem}

To understand the intuition behind the counterexample, it is helpful to consider the definition of KL divergence:
\[KL(D,D') = \ce(D,D')-H(D).\]
A small KL divergence gives an upper bound on the total variation distance $|D-D'|$ by Pinsker's inequality. 
If $H(D)$ is held constant, then relatively small changes in $\ce(D,D')$ also certify closeness in total variation distance.
But in this counterexample, we will decrease the entropy $H(D)$ by $k>1$ and, therefore, this allows us to increase the KL divergence while keeping a similar value of cross-entropy. 

\

\begin{proof}(\emph{Sketch})

The basic idea is to consider a ``rescaled'' distribution on $1/k$ of the outputs for some sufficiently large integer $k$.  
That is, we will assign probability 0 to $1-\frac{1}{k}$ fraction of the outputs, and multiply the probabilities on the remaining outputs by $k$. 
By construction, this has total variation distance roughly $1-\frac{1}{k}$ from the ideal distribution and relatively small entropy.
However, one can show it is essentially indistinguishable from the point of cross-entropy difference -- that is the cross-entropy difference is exponentially close to the ideal.

To be more precise, consider listing the strings $x \in \{0,1\}^n$ as $x_1, \ldots, x_N$ in order of increasing $p_U(x)$. 
Label the strings $x_i$, $i=1\ldots N$, such that $i<j$ implies $p_U(x_i)<p_U(x_j)$. 
For simplicity, we will focus only on the ``middle $99.9$ percent'' of the distribution, i.e. we will pick constants $c_1,c_2$ such that with high probability over the choice of $U$, $99.9$ percent of probability mass is on $x_i$ satisfying $\frac{c_1}{N}<p_U(x_i)\leq \frac{c_2}{N}$.  We will consider values of $i$ between $i_{min}$, the smallest $i$ such that $\frac{c_1}{N}<p_U(x_i)$, and $i_{max}$, the largest $i$ such that $p_U(x_i)<\frac{c_2}{N}$.

Now consider the distribution $D_U$ defined as follows: 
 \[D_U(x_i) = \begin{cases} 
 p_U(x_i) & i<i_{min} \quad \mathrm{ or }\quad i>i_{max}\\
 p_U(x_i)+p_U(x_{i+1})+\ldots+p_U(x_{i+k-1}) & i_{min}\leq i\leq i_{max} \quad \mathrm{ and } \quad i = k \mathbb{N}\\
 0                     & i_{min}\leq i\leq i_{max} \quad \mathrm{ and } \quad i \neq k\mathbb{N}.\\
\end{cases}\]

It's not hard to show see that the total variation distance between this distribution and the ideal distribution is $0.99(1-\frac{1}{k})$ in expectation over the choice of $U$, and hence if $k=500$ with high probability is more than $0.99$ by standard concentration inequalities. %
Furthermore, a careful but straightforward calculation shows that the $\ced$ of this rescaled distribution $D_U$ and $p_U$ is exponentially close to $1$, which is the ideal score.

In short, the cross-entropy difference does not behave like a metric\footnote{We note that one can create distributions which make the cross-entropy difference \emph{greater than 1} as well, by simply piling more probability mass on the "heavy" elements of the distribution and less on the ``light" elements of the distribution.}: achieving cross-entropy difference close to $1$ does not certify closeness in total variation distance.

 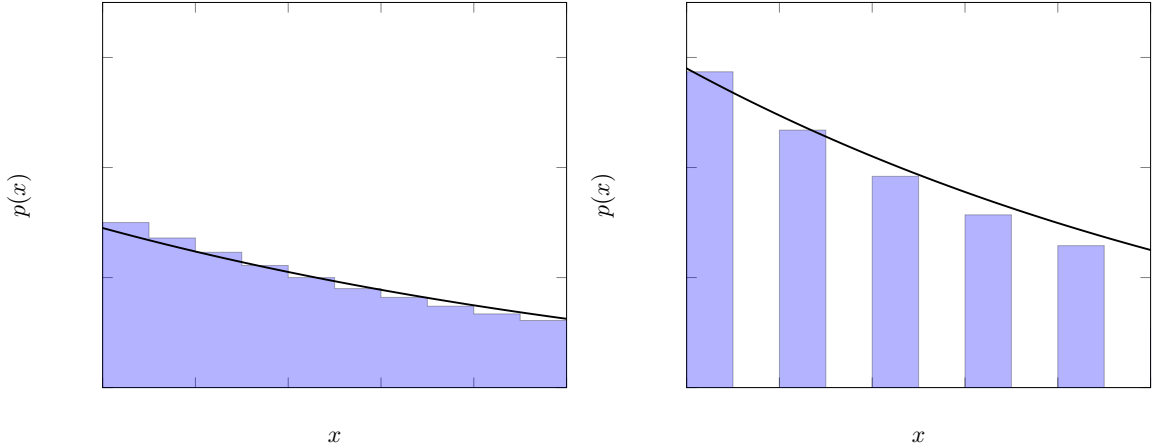
\begin{figure}[!ht]
 \centering
 \begin{tikzpicture}[scale = 0.9]
 \begin{axis}[ymin=0,ymax=0.35,xmin=0, xmax=1,enlargelimits=false,
     yticklabels={,,},xticklabels={,,},
     xlabel=$x$,ylabel=$p(x)$]
 \addplot
 	[const plot,fill=blue,draw=black,opacity=0.3] 
 coordinates
 {(0,0.15)    (0.1,0.136)  (0.2,0.123)   (0.3,0.111)
  (0.4,0.10) (0.5,0.09)  (0.6,0.082)  (0.7,0.074)
  (0.8,0.067) (0.9,0.061)  (1,0.055)} 
 	\closedcycle;
 \addplot[draw=black,thick,samples=400]{0.165 * pow(2,-x) - 0.02};
 \end{axis}
 \end{tikzpicture} \begin{tikzpicture}[scale = 0.9]
 \begin{axis}[ymin=0,ymax=0.35,,xmin=0, xmax=1,enlargelimits=false,
       yticklabels={,,},xticklabels={,,},
       xlabel=$x$,ylabel=$p(x)$]
 \addplot
 	[const plot,fill=blue,draw=black,opacity=0.3] 
 coordinates
 {(0,0.287)    (0.1,0)  (0.2,0.234)   (0.3,0.)
  (0.4,0.192) (0.5,0)  (0.6,0.157)  (0.7,0)
  (0.8,0.129) (0.9,0.)  (1,0.055)} 
 	\closedcycle;
 \addplot[draw=black,thick,samples=400]{0.33 * pow(2,-x) - 0.04};
 \end{axis}
 \end{tikzpicture}
 \caption{On the left, the initial output distribution. On the right, the ``rescaled'' distribution.}
 \end{figure}
\end{proof}

Although we have shown that cross-entropy does not generically certify total variation distance, we note that the Google/UCSB proposal makes the \emph{assumption} that their device either performs perfectly or else outputs the maximally mixed state on every run of their experiment \cite{boixo2016characterizing, boixo2017fourier}.  Equivalently, there exists an $\alpha \in [0,1]$ such that for each outcome $x \in \bits^n$,
\begin{equation}
    p_{dev}(x) = \alpha p_U(x) + (1-\alpha) \frac{1}{N}. \label{eq:assumption}
\end{equation}
Once this assumption is made, achieving cross-entropy close to the ideal implies closeness to the perfect output state in total variation distance: one can easily compute\footnote{This follows by the linearity of the measure, and the fact the uniform distribution is $1/e$-close to Porter-Thomas.} that achieving $\ced = 1-\epsilon$, together with the assumption in eq. ($\ref{eq:assumption}$) implies that $\mathbb{E}_{U\sim \mathbb{H}}|p_{dev}-p_U| \leq \frac{\epsilon}{e} \approx 0.37 \epsilon$. 
This assumption is reached via empirical evidence from their 9-qubit device \cite{neill2017blueprint} that errors cause their output distribution to look closer to uniform, as well as through numerical simulations of how an ideal device should behave under a particular noise model \cite{boixo2016characterizing}.  However, a 49 qubit device will likely be too large to verify this assumption.

\subsection{``Shape'' does not verify the ``quantumness'' of the output distribution}

Since the above example rules out a strong correlation between cross-entropy and total variation distance,
it is natural to wonder if some other property of outcome distributions arising from Random Circuit Sampling experiments could be put forward as a basis for a verifiable quantum supremacy proposal.  

An often mentioned candidate property is the density of probabilities in the outcome distribution.  The suggestion is that one can verify the ``quantumness'' of a system simply by analyzing the ``shape'' of the outcome distribution. A key property of typical distributions drawn from $P$ is that they will have a ``Porter-Thomas shape'' (recall, $P$ is the joint distribution over all $N$ output probabilities generated by choosing a Haar random $U$).   
That is, if one draws a vector $v \sim P$, then for any choice of constants $c_1<c_2$ the number of $x$ with $v_x$ in the range $[c_1/N, c_2/N]$ will be roughly $N\int_{c_1}^{c_2} e^{-q} dq$ in expectation over the choice of $v$ (i.e. the choice of unitary $U$).
Therefore, by concentration of measure, with high probability over the choice of $v$ from $P$, the distribution induced by choosing a random $x$ and sampling $v_x$ is close to (a discretized version of) Porter-Thomas.
Indeed, in the Google/UCSB proposal such a ``shape'' is referred to as a ``signature'' of quantum effects (see e.g., page 3 of \cite{boixo2016characterizing}).  

Note that since the Porter-Thomas distribution has an analytic description, there is a trivial classical algorithm for sampling from it. 
The more interesting question is whether any classical \emph{physical} processes can  reproduce the ``Porter-Thomas shape'', and how well these processes could score in cross-entropy.  We give an example of a simple classical physical process which produces probability distributions which are approximately Porter-Thomas in shape.
Furthermore, the classical process resembles the physics of a noisy/decoherent quantum system.  Consequently, the exponential nature of the Porter-Thomas distribution is not a signature of ``quantumness.''

In particular, consider a system of $n+m$ classical bits, the first $n$ of which we will call the ``system'', and the second $m$ of which we will call the ``environment''. 
Suppose that the system bits are initialized to $0$, while the environment bits are chosen uniformly at random.  Now suppose that one applies a uniformly random classical permutation to these $n+m$ bit strings (i.e. a random element $\sigma$ of $S_{2^{n+m}}$) and observes the first $n$ system bits many times (while ignoring the environment bits) with the same choice of $\sigma$ but different settings of the environment bits. 
A diagram of this process is provided below in quantum circuit notation, but note this is a purely classical process.
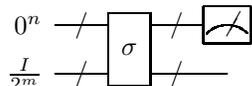
\[
    \Qcircuit @C=1em @R=0.7em {
\lstick{0^n}  & {/}\qw & \multigate{1}{\sigma} & {/}\qw& \meter \\
\lstick{\frac{I}{2^m}}  & {/} \qw &\ghost{\sigma} & {/}\qw & \qw
}
\]
A natural question is what ``shape'' of probability distribution does this process produce? 
Over the choice of $\sigma$, each input string on $n+m$ bits is mapped to a uniformly random output string on $n+m$ bits (of which we only observe the first $n$ bits). 
Therefore, this process resembles throwing $2^m$ balls (one for each possible setting of the environment bits) into $2^n$ bins (one for each possible output string of the system bits).
This is not exactly the process performed because each ball is not thrown independently due to the fact that $\sigma$ is a permutation rather than a random function.
However, if $m$ is sufficiently large -- say if $m=n$ -- then the value of each string is \emph{approximately} independent.
This is because we are only observing the first $n$ bits of the output string -- therefore, each bin we observe consists of $2^m$ possible output strings. 
So the fact that one string mapped to a particular observed output only very slightly decreases the probability another string does so.

Therefore, this classical system is well approximated by the process of throwing $2^m$ balls into $2^n$ bins. 
For simplicity, suppose we set $m=n$ (though we do not claim this choice is optimal). 
It is well known that in the large $n$ limit, the distribution of the number of balls in each bin is close to the Poisson distribution with mean $2^{m-n}=1$ \cite{motwani2010randomized}. 
We note that this process is still approximately Poisson if $\sigma$ is chosen $k$-wise independently (rather than truly uniformly random) for sufficiently large $k=\poly(n)$, since the number of bins with $k$ balls is a $k$th order moment of the distribution, and in the Poisson distribution with mean $1$, almost all bins will contain  $<\poly(n)$ balls with high probability.

Then this approximately produces a Poisson distribution with mean 1, i.e. the number of balls thrown into each bin is described as:
\[\mathrm{Pr}[c = k] = \frac{1}{k! e}\]
where $c$ is the count in a particular bin.
Now to better match the parameters in the Porter-Thomas distribution, we will consider normalization by the number of balls.  Letting $N=2^n$, we see that for any output string $x$,
\[\mathrm{Pr}_x\left[p_{\text{\tiny Poisson}}(x) = \frac{k}{N}\right] = \frac{1}{k! e}.\]

We claim that this distribution is a natural classical imposter of Porter-Thomas. 
Since $k! = 2^{\Theta(k\log k)}$, this distribution is also (approximately) exponential.
So this can be seen as a discretized version of Porter-Thomas, where the discretization resolution can be made finer by choosing larger $m$.
Just as the Porter-Thomas distribution approximately describes the distribution on output probabilities of a quantum system under a random choice of $U$, here the Poisson distribution approximately describes the distribution on output probabilities of this classical system under a random choice of $\sigma$. And as the Porter-Thomas distribution is reproduced with unitary $k$-designs for sufficiently large $k$, here the Poisson statistics are reproduced when $\sigma$ is chosen from a $k$-wise independent family for sufficiently large $k$. 

This shows that Random Circuit Sampling cannot be verified purely by the shape, or probability density, of the outcome distributions. This means that any supremacy proposal based on outcome statistics must directly incorporate the \emph{relationship} between outcome strings and their probabilities. This relationship is addressed by cross-entropy difference because in order to compute this, one must compute the ideal output probabilities of the experimentally observed samples $x$.

\subsection{The relationship between cross-entropy and Heavy Output Generation}

In this section, we will discuss a recent proposal of Aaronson and Chen \cite{aaronson2016complexity} and how it relates to cross-entropy.

In the Aaronson and Chen proposal, the task required of the quantum computer is relatively simple: given a circuit description of a unitary $U$, output random strings $x_1\dots x_k$ such that at least 2/3 of them or more are above-median weight in the distribution $p_U$. In other words, most of the samples output by the experimental device should be ``heavy''.
The proposal seems to directly address the relationship between outcome strings and their probabilities.
Here we restate this proposal in the language of cross-entropy to facilitate comparison and highlight their similarity:

\newcommand{\HoG}{\mathsf{HOG}}

\begin{definition}[\cite{aaronson2016complexity}]
A family of distributions $\{D_U\}$ satisfies \emph{Heavy Output Generation (HOG)} iff the following holds: Let
\[
\HoG(D_U,p_U) = \sum_{x \in \bits^n} D_U(x) \delta(p_U(x))
\]
where $\delta(z) = 1$ if $z \geq \frac{\ln 2}{N}$ and $0$ otherwise. 
Then the family is said to satisfy HOG if 
\[\mathbb{E}_{U \sim \Hh} \HoG(D_U,p_U) \geq 2/3.
\] 
\end{definition}
The quantity $\ln(2)/N$ is chosen because it is the median of Porter-Thomas. This is empirically measured as follows: pick a random $U$, obtain $k$ samples $x_1, \ldots, x_k$ from the experimental device and compute
\begin{equation}
H = \frac{1}{k} \sum_{i = 1, \ldots, k} \delta(p_U(x_i)).
\label{eq:measuring-hog}
\end{equation}
Analagous to the case of cross-entropy, only a small number of samples will be required to get a good estimate of $H$ by concentration of measure.
Note, that the ideal device will satisfy HOG since \[\mathbb{E}_U (\HoG(p_U,p_U)) = \frac{1 + \ln 2}{2} \approx 0.85.\] Therefore, there is some tolerance for error in the experiment.

Notice the similarities between cross-entropy and HOG (eqs. (\ref{eq:measuring-ce}) and (\ref{eq:measuring-hog})): Both are approximating the expectation value of some function of the ideal output probabilities $f(p_U(x_i))$ over the experimental output distribution. In the case of cross-entropy, $f(x) = \log(1/x)$. And in the case of HOG, $f(x) = \delta(x)$. 
Both measures are constructed such that for small system sizes, a supercomputer can be used to verify a solution to either measure by computing the output probabilities $p_U(x_i)$.

Just as achieving a high-cross entropy score does not imply closeness to the ideal distribution in total variation distance (Appendix \ref{sec:cross-entropy-neq-tv}), achieving a high score on HOG does not imply closeness in total variation distance either\footnote{An experimental device could always output the heaviest item and score well on HOG while being far in total variation distance from ideal.}. Both of these measures are projecting the output distribution of the experimental device (which lives in a very high dimensional space) onto a 1-dimensional space and using this as a proxy for supremacy. We observe that these are two distinct measures as they are capturing different one-dimensional projections.

\end{document}